\documentclass[journal]{IEEEtran}
\usepackage{amsmath,amssymb,amsthm,amsfonts}
\usepackage{enumitem}  % for using the left marging in itemize
\usepackage[
backend=bibtex,
style=numeric,
giveninits=true,
sorting=nyt
]{biblatex}
\usepackage{graphicx}
\usepackage{float}
\usepackage{subcaption}
\usepackage{booktabs}
\usepackage{multirow}
\usepackage{algorithm}
\usepackage{algpseudocode}
\usepackage[hidelinks,pdfauthor={Hassan Ou-azzou, Imad Kissami, El Houcine Bergou}]{hyperref}
\hypersetup{
	colorlinks   = true,
	citecolor    = blue
}
\usepackage{relsize}
\usepackage{cleveref}
\usepackage{comment}
\usepackage{color,xcolor,soul}

\soulregister\cite7
\soulregister\ref7
\theoremstyle{plain}
\newtheorem{theorem}{Theorem}[section] % Main counter
    \newtheorem{lemma}[theorem]{Lemma} % Shares numbering with theorem
    \newtheorem{proposition}[theorem]{Proposition}
    \newtheorem{corollary}[theorem]{Corollary}
    \newtheorem{definition}[theorem]{Definition}
    \newtheorem{example}[theorem]{Example}
    \newtheorem{remark}[theorem]{Remark}
    
    \newtheorem{assumption}[theorem]{Assumption}

\DeclareMathOperator{\nnzr}{nnzr}

\DeclareMathOperator{\Circ}{Circ}

\newcommand{\one}{\mathbf{1}}

\title{Binary Cyclic Gradient Coding: Exact Recovery and Robustness to Compression}

\author{Hassan Ou-azzou, Imad Kissami, and El Houcine Bergou%
\thanks{The authors are with the College of Computing, Mohammed~VI Polytechnic University, Ben Guerir, Morocco (e-mail: hassan.ouazzou-ext@um6p.ma; imad.kissami@um6p.ma; elhoucine.bergou@um6p.ma).}}

\begin{document}
\maketitle

\begin{abstract}
This paper addresses the gradient coding problem in distributed optimization and coded computing. We introduce binary cyclic gradient coding (BCGC),  an exact gradient coding scheme for $n$ workers and up to $s$ stragglers,  with a binary circulant encoding matrix when $(s+1)\mid n$ and a nearly circulant matrix otherwise, enabling decoding by a single summation. Like FRC and GFRC, BCGC partitions workers into $s+1$ classes, but uses cyclic data assignments and is not permutation-equivalent to either scheme. We show that binary coding under compression is equivalent to uncoded compression on coarser data partitions, preserving convergence guarantees, while MDS codes can amplify compression errors through large decoding coefficients. Under compression, BCGC yields $s+1$ possible reconstructions when decoding the compressed gradients, against one for FRC and two for GFRC, and can select the one closest to the exact gradient, a choice unavailable to FRC and GFRC. Experiments on four datasets show that BCGC matches GFRC, remains stable under top-$\kappa$ sparsification while MDS codes diverge, and improves accuracy on heterogeneous data by reducing compression error.
\end{abstract}

\begin{IEEEkeywords}
Gradient coding, coded distributed computing, straggler mitigation, fractional repetition codes, cyclic codes, MDS codes, gradient compression.
\end{IEEEkeywords}

\section{Introduction} \label{Sec_Intro}
Distributed gradient descent splits the computation of gradients for large learning problems across several workers coordinated by a central master \cite{Bottou2018}. At large scale, its efficiency is mainly limited by two factors. The first is the communication of high dimensional gradients to the master. The second is the presence of stragglers, that is, workers that respond slowly or fail to respond.

During the process each worker sends a high dimensional real vector to the master at every iteration over bandwidth limited links \cite{NgSurvey2021,Xu2021GRACE}. This cost can be  reduced by compression, mainly quantization \cite{Alistarh2017QSGD}, sparsification \cite{Alistarh2018Topk,Lin2018DGC}, sketch based quantization \cite{Jiang2018SketchML} and low rank approximation \cite{Vogels2019PowerSGD}, often with error feedback, which reuses the discarded part of each message in later iterations \cite{Stich2018,Karimireddy2019EF,Richtarik2021EF21}. Stragglers are handled either by classical methods that add redundancy or move work at run time, such as task replication \cite{Wang2015Replication}, work stealing \cite{Blumofe1999} and work exchange \cite{Attia2017WorkExchange}, or by coding theory, which builds the redundancy into the data assignment before the computation, giving coded distributed computing (CDC) \cite{NgSurvey2021}. The master then recovers the result from the fastest workers with a worst case straggler guarantee, at the price of a higher computation load. CDC also reduces communication by encoding intermediate results \cite{Li2018Tradeoff} and offers fault tolerance, privacy \cite{Nodehi2018Privacy} and security \cite{Yu2018Lagrange}.

In this paper, we study how gradient coding interacts with the compression of the coded messages: the coding techniques handle the stragglers, and each worker compresses its coded message to reduce communication. Gradient coding, introduced by Tandon et al.\ \cite{Tandon2017}, partitions the training data into replicated chunks assigned to the workers by an encoding matrix, so that a decoder recovers the exact full gradient from the fastest workers, without knowing in advance which respond. A parallel line of work relaxes exact recovery, through sparse random graphs \cite{Charles2017}, block designs for adversarial stragglers \cite{Kadhe2019}, stochastic assignments \cite{Bitar2020}, approximate codes with optimal decoding \cite{Glasgow2021,Wang2019erasure},  and online adaptation to the straggler statistics \cite{KasHanna2020,KasHanna2022}. In addition others adapt the assignment to the observed stragglers \cite{GCDC2022}, and  also the idea is used to  speed up distributed matrix computations \cite{Lee2018,Dutta2020MatDot}.

Standard gradient coding does not reduce communication, since each responding worker sends a vector of dimension $d$ at every iteration. It can be reduced by increasing the computation load so that each worker sends a shorter coded vector \cite{Ye2018,Kadhe2020CommEff}, or by compressing the coded messages \cite{Li2024OneBit,Li2026Biased,SongChoi2026,Rudow2023Compression}. Compressing coded messages raises two difficulties. First, coded messages are linear combinations of partial gradients, whereas compression is generally nonlinear and can destroy the linear structure that decoding requires. Existing works handle this in specific settings, around a one bit quantizer \cite{Li2024OneBit}, biased compressors with error feedback \cite{Li2026Biased}, jointly designed approximate codes and quantization \cite{SongChoi2026}, or a Byzantine threat model \cite{LiCyclic2026}. Second, exact recovery constructions with minimal storage, such as the Gaussian MDS code of \cite{Tandon2017}, cyclic MDS codes \cite{Raviv2019}, Reed--Solomon codes \cite{Halbawi2018} and LDPC based schemes \cite{Maity2019LDPC}, decode with real or complex coefficients, which amplify rounding errors as the number of workers grows \cite{Charalambides2020Stable,RamamoorthyTang2022} and, under compression, amplify the compression errors (Proposition~\ref{prop:diff}). Binary gradient coding with decoding by summation removes the second difficulty. This is the principle of the fractional repetition coding (FRC) scheme of \cite{Tandon2017}, named after the fractional repetition codes for distributed storage \cite{ElRouayhebRamchandran2010}, and of its generalization GFRC \cite{Charalambides2025}, where the encoding matrix is binary and  each worker sends the sum of its partial gradients, so the master recovers the full gradient by summing the fastest workers without solving a linear system. 

We propose \emph{binary cyclic gradient coding} (BCGC), an exact scheme whose encoding matrix is binary and circulant, or nearly circulant when $(s+1)\nmid n$, and whose decoding reduces to a single summation. Here $n$ is the number of workers, indexed by $[n]:=\{0,\ldots,n-1\}$, $k$ the number of chunks, and $s$ the number of stragglers to tolerate. Following \cite{Charalambides2025}, we set $k=n$ and split the workers into $s+1$ classes by their index modulo $s+1$. Each worker receives a cyclic support of $s+1$ consecutive indices modulo $n$, the last worker of each class taking a support of adjusted length when $(s+1)\nmid n$. The workers of each class cover $[n]$ exactly once, so summing any one class gives the full gradient; since $s$ stragglers cannot meet all $s+1$ classes, at least one is always complete, and the scheme tolerates any $s$ stragglers at the recovery threshold $n-s$, which cannot be lowered, for every $1\le s\le n-1$ (the case $s=0$ being trivial, with $\mathbf{B}=I_n$).

Several features distinguish our construction. Like FRC and GFRC, and unlike cyclic MDS \cite{Raviv2019} and Reed--Solomon \cite{Halbawi2018} codes, BCGC is binary and decodes by summing one complete class, identified online in $O(1)$ per response as in GFRC \cite{Charalambides2025}, but its binary recovering sets are exactly its classes, whereas FRC and GFRC may combine workers of different classes (Remark~\ref{rem:cooperation-div}). BCGC is not permutation-equivalent to FRC or GFRC and attains the minimal total load deviation of GFRC with a larger maximum load. The same circulant support appears in \cite{Tandon2017} with real coefficients and in \cite{LiCyclic2026} with randomly reassigned rows and approximate aggregation, whereas here the assignment is fixed and the decoding exact; existing cyclic gradient codes are either real valued, hence numerically unstable and sensitive to compression, or binary with approximate decoding. Under a nonlinear compressor, binary decoding turns the decoded update into an uncoded compressed update on the partition formed by the supports of the decoded class, and for contractive compressors FRC, GFRC and BCGC share the standard guarantee of compressed gradient descent and differ only in these partitions.  The cyclic supports produce $s+1$ distinct partitions when $n\ge 2(s+1)$, against one for FRC and two for GFRC, so BCGC can select the reconstruction closest to the exact gradient, a choice unavailable to FRC and GFRC (Section~\ref{Sec_Compression}). Our comparison is with the standard block assignments. The general GFRC framework \cite{Charalambides2025} also permits per-class chunk permutations, so partition diversity is not excluded for every GFRC variant, but our explicit construction achieves it with a fixed assignment and exact decoding.

To summarize, our contributions are as follows.
\begin{itemize}
\item \textbf{Construction and exact recovery.} BCGC assigns the data through cyclic supports, giving a circulant encoding matrix when $(s+1)\mid n$ and a nearly circulant one otherwise, and decodes by a single summation over one complete class,  it tolerates any $s$ stragglers at the optimal threshold $n-s$ with the minimal total load $n(s+1)$.
\item \textbf{Algebraic structure.} The code spanned by the rows of the encoding matrix of the BCGC scheme has dimension $n-s$ and is cyclic when $(s+1)\mid n$; for $2\le s\le n-2$ it is not MDS, and the matrix is not permutation-equivalent to those of FRC and GFRC.
\item \textbf{Load distribution.} BCGC attains the minimal total load deviation of GFRC, keeping the load $s+1$ on $n-s-1$ workers, with every worker and chunk at load $s+1$ when $(s+1)\mid n$.
\item \textbf{Compatibility with compression.} First we prove  that any compressed binary scheme is exactly the uncoded compressed scheme on a coarser partition of the data, whose blocks are the supports of the decoded class. This keeps binary schemes stable, whereas the large decoding coefficients of MDS codes amplify the compression error. For every contractive compressor the three schemes share the same convergence guarantee, but BCGC induces $s+1$ distinct partitions (none refining another), against one for FRC and two for GFRC, so one of its classes can be exact when those of FRC and GFRC are not.
\item \textbf{Experiments.} On the Toubkal supercomputer with $n=100$ workers and four datasets (Amazon, MNIST, EMNIST and 20~Newsgroups), BCGC matches GFRC in robustness and accuracy and is faster than the Gaussian MDS and Reed--Solomon codes. Under top-$\kappa$ sparsification with $\kappa=\lfloor0.05\,d\rfloor$ it stays within $0.022$ of its uncompressed metric while those codes diverge, and on heterogeneous data, selecting the complete class with the smallest compression error improves accuracy over GFRC by up to $5$ points at the same cost.
\end{itemize}

The paper is organized as follows. Section~\ref{Sec_Prelim} recalls gradient coding and its binary variant. Section~\ref{Sec_divisible} presents BCGC in both the divisible and non-divisible cases, with its recovery guarantee, its algebraic structure, and its comparison with FRC and GFRC. Section~\ref{Sec_Compression} studies the effect of compression on binary gradient codes such as FRC, GFRC and BCGC. Section~\ref{Sec_Experiments} reports the experimental results, and Section~\ref{Sec_Conclusion} concludes the paper.

\section{Preliminaries}\label{Sec_Prelim}

In this section, we recall the general framework of gradient coding \cite{Tandon2017}, together with its binary variant, which will be used throughout the paper. For a positive integer $m$, we write $[m]:=\{0,1,\ldots,m-1\}$.

\subsection{Gradient coding}
Consider a distributed learning setting in which a master solves an optimization problem by gradient descent, with the computation distributed across $n$ workers. Given a dataset $\mathcal{D}$ of $N$ samples $(x,y)$, the goal is to learn a parameter $\boldsymbol{\beta}^*\in\mathbb{R}^{ d}$ minimizing the loss function
\[
L(\mathcal{D},\boldsymbol{\beta})=\sum_{(x,y)\in\mathcal{D}}\phi(x,y;\boldsymbol{\beta}),
\]
where $\phi$ is the per-sample loss.  We solve
\begin{equation}
\label{eq:erm}
\boldsymbol{\beta}^* = \arg\min_{\boldsymbol{\beta}\in\mathbb{R}^{ d}}
L(\mathcal{D}, \boldsymbol{\beta})
\end{equation}
by gradient descent, iterating the parameter update
\begin{equation}
\label{eq:update_param}
\boldsymbol{\beta}_{t+1} = \boldsymbol{\beta}_t - \gamma \, \mathbf{g},
\end{equation}
where $\gamma > 0$ is the learning rate, kept constant throughout the paper, and $\mathbf{g} = \nabla L(\mathcal{D}, \boldsymbol{\beta}_t)$ is the full gradient. Throughout, parameters and gradients are regarded as row vectors in $\mathbb{R}^{ d}$. To distribute this computation, we partition $\mathcal{D}$ into $k$ chunks $\mathcal{D}_0,\dots,\mathcal{D}_{k-1}$, of equal size in our experiments. We write the \emph{partial gradients} $\mathbf{g}_j := \nabla L(\mathcal{D}_j, \boldsymbol{\beta}_t)$, so that the full gradient decomposes as
\begin{equation}
\label{eq:full_grad}
\mathbf{g} = \sum_{j=0}^{k-1} \nabla L(\mathcal{D}_j, \boldsymbol{\beta}_t)
= \sum_{j=0}^{k-1} \mathbf{g}_j.
\end{equation}
The $n$ workers compute the partial gradients $\mathbf{g}_j$ where at most   $s$ of them are  \emph{stragglers} at each iteration. The master thus receives at least $n-s$ responses, indexed by a set $I$ of \emph{fastest workers}, and must recover $\mathbf{g}$ exactly from them. The \emph{recovery threshold} of a scheme is the smallest $K$ such that $\mathbf{g}$ can be recovered from any $K$ responses by some decoding function. A scheme tolerating any $s$ stragglers has threshold at most $n-s$.

To recover $\mathbf{g}$ from any $n-s$ responses, we construct an encoding matrix $\mathbf{B}\in\mathbb{R}^{n\times k}$, with rows $\mathbf{b}_0,\ldots,\mathbf{b}_{n-1}$, such that, for every set $I\subseteq[n]$ with $|I|\ge n-s$, there exists a \emph{decoding vector} $\mathbf{a}_I\in\mathbb{R}^{1\times n}$, with $\operatorname{supp}(\mathbf{a}_I)\subseteq I$, satisfying
\begin{equation}\label{eq:decode-condition}
\mathbf{a}_I \mathbf{B}=\one_{1\times k}.
\end{equation}
Here $\operatorname{supp}(\mathbf{x})$ denotes the set of indices of the nonzero entries of $\mathbf{x}$. The master sends worker $i$ the chunks indexed by $\operatorname{supp}(\mathbf{b}_i)$ with the coefficients $b_{i,j}$. At each iteration $t$ it broadcasts $\boldsymbol{\beta}_t$, and worker $i$ returns the encoded message
$$
\widehat{\mathbf{g}}_i:=\sum_{j\in\operatorname{supp}(\mathbf{b}_i)} b_{i,j}\,\mathbf{g}_j = \mathbf{b}_i \bar{\mathbf{g}},\qquad \bar{\mathbf{g}} = [\mathbf{g}_0^\top,\ldots,\mathbf{g}_{k-1}^\top]^\top.
$$
Once any set $I$ with $|I|\ge n-s$ has responded, the master recovers
\begin{equation}\label{eq:decode}
\sum_{i\in I}\mathbf{a}_I(i)\,\widehat{\mathbf{g}}_i=\mathbf{a}_I\mathbf{B}\bar{\mathbf{g}}=\one_{1\times k}\bar{\mathbf{g}}=\sum_{j=0}^{k-1}\mathbf{g}_j=\mathbf{g},
\end{equation}
and performs the update~\eqref{eq:update_param}.

\noindent The condition $\mathbf{a}_I \mathbf{B}=\one_{1\times k}$ is central to gradient coding. In \cite{Tandon2017}, this condition is characterized by the $\mathbf{B}$-span condition, which we recall in the following definition.

\begin{definition}[{\cite[Condition~1]{Tandon2017}}] \label{B-span}
A matrix $\mathbf{B}\in\mathbb{R}^{n\times k}$ satisfies the \emph{$\mathbf{B}$-span condition} for $s$ stragglers if
\[
\one_{1\times k} \in \operatorname{span}\{\mathbf{b}_i\mid i\in I\}, \qquad \forall\,I\subseteq[n],\ |I|=n-s .
\]
\end{definition}

\begin{remark}
In the literature, the $\mathbf{B}$-span condition is satisfied by several schemes: fractional 
repetition codes \cite{Tandon2017}, generalized fractional repetition codes~\cite{Charalambides2025}, cyclic MDS codes~\cite{Raviv2019}, and
Reed--Solomon codes~\cite{Halbawi2018}. These schemes encode the partial gradients and recover the gradient exactly. 
\end{remark}

\subsection{Binary gradient coding}\label{Sec_binary-gc}

We now recall the \emph{binary} gradient coding schemes of \cite{Tandon2017,Charalambides2025}, in which the encoding matrix is required to satisfy $\mathbf{B}\in\{0,1\}^{n\times k}$ and, for every $I\subseteq[n]$ with $|I|=n-s$, the decoding vector is required to lie in $\{0,1\}^{1\times n}$ with $\operatorname{supp}(\mathbf{a}_I)\subseteq I$. Restricting both to $\{0,1\}$ removes scalar multiplications from the protocol entirely. Each worker sums the partial gradients it has been assigned,
\[
\widehat{\mathbf{g}}_i=\mathbf{b}_i\,\bar{\mathbf{g}}
=\sum_{j\in\operatorname{supp}(\mathbf{b}_i)} \mathbf{g}_j ,
\]
and the master sums a subset of the received messages, $\mathbf{g}=\sum_{i\in\operatorname{supp}(\mathbf{a}_I)}\widehat{\mathbf{g}}_i$, with no linear system to solve online, avoiding the error amplification of large decoding coefficients. The identity $\mathbf{a}_I\mathbf{B}=\one_{1\times k}$ is understood over $\mathbb{R}$, so \emph{binary} refers to the $0/1$ entries of $\mathbf{B}$ and $\mathbf{a}_I$, not to arithmetic modulo~$2$.

A binary matrix satisfying the $\mathbf{B}$-span condition (Definition~\ref{B-span}) need not admit binary decoding vectors, since that condition only places $\one_{1\times k}$ in the \emph{real} span of $\{\mathbf{b}_i \mid i\in I\}$. The following sufficient condition, due to Charalambides et al.\ \cite{Charalambides2025}, guarantees binary decoding by exhibiting the decoding vector explicitly. For $K\subseteq[n]$, we write $\one_K\in\{0,1\}^{1\times n}$ for the indicator row vector of $K$.

\begin{proposition}[{\cite[Prop.~3, Cor.~4]{Charalambides2025}}]
\label{prop:partition-sum}
Let $\mathbf{B}\in\{0,1\}^{n\times k}$ and suppose that $[n]$ admits a partition into $s+1$ nonempty classes $K_0,\ldots,K_s$ such that
\begin{equation}\label{eq:class-cover}
\sum_{i\in K_j}\mathbf{b}_i=\one_{1\times k}, \qquad j=0,\ldots,s .
\end{equation}
Then, for every $I\subseteq[n]$ with $|I|=n-s$, there exists $j$ with $K_j\subseteq I$, and $\mathbf{a}_I:=\one_{K_j}$ satisfies $\mathbf{a}_I \mathbf{B}=\one_{1\times k}$.
\end{proposition}

We now recall the encoding of the \emph{fractional repetition code} (FRC) of Tandon et al.~\cite{Tandon2017}, a binary scheme in which each worker receives a fraction of the data, repeated across the $s+1$ classes. Assume $(s+1)\mid n$ and $k=n$, and write $\ell:=\tfrac{n}{s+1}$. A single class is encoded by the block matrix
\begin{equation}\label{eq:Bblock}
\mathbf{B}_{\text{block}}=I_{\ell}\otimes\one_{1\times(s+1)}\in\{0,1\}^{\ell\times n},
\end{equation}
whose $\ell$ rows have disjoint supports of size $s+1$ covering $[n]$. Stacking $s+1$ copies gives
\begin{equation}\label{eq:Bfrac}
\mathbf{B}_{\mathrm{FRC}}=
\begin{bmatrix}
\mathbf{B}_{\text{block}}\\
\vdots\\
\mathbf{B}_{\text{block}}
\end{bmatrix}
= \one_{(s+1)\times 1}\otimes \mathbf{B}_{\text{block}}\in\{0,1\}^{n\times n},
\end{equation}
whose classes are the consecutive blocks of $\ell$ rows. Each such class is a copy of $\mathbf{B}_{\text{block}}$ and therefore satisfies \eqref{eq:class-cover}, so Proposition~\ref{prop:partition-sum} applies and $\mathbf{B}_{\mathrm{FRC}}$ tolerates $s$ stragglers. Renumbering the worker $q$ of the class $j$ as $q(s+1)+j$ yields the permutation-equivalent encoding
\begin{equation}\label{eq:BFRCprime}
\mathbf{B}'_{\mathrm{FRC}}=I_{\ell}\otimes\one_{(s+1)\times(s+1)}\in\{0,1\}^{n\times n},
\end{equation}
which is block diagonal with $\ell$ all-ones blocks of size $(s+1)\times(s+1)$, and in which the classes are the congruence classes $K_j=\{\,i\in[n] \mid i\equiv j \pmod{s+1}\,\}$, each selecting one row from every diagonal block.

\section{Binary Cyclic Gradient Coding}
\label{Sec_divisible}

We now present the \emph{binary cyclic gradient coding} (BCGC) scheme, on which our results are based. BCGC combines two elements, the cyclic structure of cyclic codes and the binary encoding of FRC \cite{Tandon2017} and of its generalization GFRC \cite{Charalambides2025}. Its supports are those of the cyclic repetition scheme of \cite{Tandon2017}, but with binary coefficients instead of real ones and with decoding by the summation of one class instead of a linear system, and it does not use the cyclic MDS codes of \cite{Raviv2019}. As in FRC and GFRC, the workers are partitioned into $s+1$ classes, each of which collectively holds the entire dataset, so that summing the messages of any complete class recovers the full gradient. BCGC differs from FRC and GFRC in the choice of the supports: its cyclic supports make the encoding matrix circulant when $(s+1)\mid n$ and near-circulant otherwise, and for this choice we characterize exactly the binary recovering sets.

Following \cite{Tandon2017,Charalambides2025}, we set $k=n$ and index both the workers and the data chunks by $[n]$, so that the partial gradients are $\mathbf{g}_0,\ldots,\mathbf{g}_{n-1}$ and $\mathbf{g}=\sum_{j\in[n]}\mathbf{g}_j$. We write the Euclidean division of $n$ by $s+1$ as
\[
n=\ell(s+1)+r,\qquad \ell=\left\lfloor\frac{n}{s+1}\right\rfloor\ge1,\qquad 0\le r\le s .
\]
The \emph{divisible case} is $r=0$, that is $(s+1)\mid n$, and the \emph{non-divisible case} is $1\le r\le s$. The case $s=0$ is trivial, with $\mathbf{B}=I_n$.

\subsection{The encoding scheme}
As in \cite{Charalambides2025}, we partition the workers according to their indices modulo $s+1$. For $j\in\{0,\ldots,s\}$, the \emph{congruence class} of $j$ is
\[
K_j:=\{\,i\in[n]\mid i\equiv j\pmod{s+1}\,\}.
\]
We prove in Theorem~\ref{thm:char} that these classes are precisely the binary recovering sets of BCGC. Their sizes come in two values, as follows.

\begin{lemma}[Class sizes {\cite[eq.~(12)]{Charalambides2025}}]\label{lem:class-size}
For $j=0,\ldots,s$, $K_j=\{\,j+q(s+1)\mid 0\le q\le |K_j|-1\,\}$, with
\[
|K_j|=\ell+1\ \text{ if }\ 0\le j<r,\qquad |K_j|=\ell\ \text{ if }\ r\le j\le s .
\]
When $r=0$, all classes have size $\ell$. Moreover, $[n]=K_0\sqcup\cdots\sqcup K_s$.
\end{lemma}

\begin{proof}
\begin{enumerate}
\item \emph{Size of $K_j$.} An element $i\in[n]$ lies in $K_j$ exactly when $i=j+q(s+1)$ with $0\le q\le (n-1-j)/(s+1)$, so
\[
|K_j|=\left\lfloor\frac{n-1-j}{s+1}\right\rfloor+1 .
\]
Writing $n-1-j=(\ell-1)(s+1)+(r+s-j)$ gives
\[
\left\lfloor\frac{n-1-j}{s+1}\right\rfloor=\ell-1+\left\lfloor\frac{r+s-j}{s+1}\right\rfloor .
\]
If $0\le j<r$, then $s+1\le r+s-j<2(s+1)$, so the last floor equals $1$ and $|K_j|=\ell+1$. If $r\le j\le s$, then $0\le r+s-j<s+1$, so the last floor equals $0$ and $|K_j|=\ell$.
\item \emph{Partition.} The classes are pairwise disjoint, since $0,\ldots,s$ are distinct modulo $s+1$, and their sizes add up to
\[
r(\ell+1)+(s+1-r)\ell=n ,
\]
so $[n]=K_0\sqcup\cdots\sqcup K_s$.
\end{enumerate}
\end{proof}

\noindent In the divisible case the classes are, moreover, cyclic shifts of one another, which underlies the circulant structure of $\mathbf{B}$ and the shift relation between the partitions used later.

\begin{proposition}\label{Prop_Kshift}
If $r=0$, then $K_{j+1}=K_j+1$ for $j=0,\ldots,s-1$, where $K_j+1:=\{\,(x+1)\bmod n\mid x\in K_j\,\}$.
\end{proposition}

\begin{proof}
Identify $[n]$ with $\mathbb Z_n$. Since $r=0$, the set $H:=\langle s+1\rangle$ is a subgroup of $\mathbb Z_n$ of order $\ell$, and $K_j=j+H$ for every $j$. Therefore $K_j+1=(j+1)+H=K_{j+1}$.
\end{proof}

\noindent We now define the cyclic support $S_i$ of worker $i$, that is, the set of indices of the data chunks assigned to it. Within each class, the first workers receive a full cyclic support of size $s+1$, and the last worker receives the remaining indices.

\begin{definition}[Cyclic support]\label{Def_cyclic-support}
Let $j\in\{0,\ldots,s\}$, $m_j:=|K_j|$, and let $i_q=j+q(s+1)$, $0\le q\le m_j-1$, be the workers of $K_j$. The \emph{cyclic support} of worker $i_q$ is
\[
S_{i_q}:=\{\,(i_q+t)\bmod n\mid 0\le t\le L_{i_q}-1\,\},
\]
where
\[
L_{i_q}:=\begin{cases} s+1, & 0\le q\le m_j-2,\\ n-(m_j-1)(s+1), & q=m_j-1.\end{cases}
\]
Thus every worker receives $s+1$ consecutive indices, except the last worker of each class. When $r=0$, every $L_{i_q}=s+1$.
\end{definition}

\noindent We take as rows of $\mathbf{B}$ the indicator vectors $\mathbf{b}_i:=\one_{S_i}\in\{0,1\}^{1\times n}$. Since $L_i\le s+1+r\le n$, the elements $(i+t)\bmod n$, $0\le t\le L_i-1$, are pairwise distinct, $|S_i|=L_i$, and
\begin{equation}\label{eq:indicator}
b_{i,l}=1\iff l\in S_i\iff(l-i)\bmod n\le L_i-1 .
\end{equation}
The last worker of $K_j$ is $n-r+j$ if $j<r$ and $n-s-1-r+j$ if $j\ge r$, so the sizes of the cyclic supports are
\begin{equation}\label{eq:nd-Li}
L_i=
\begin{cases}
s+1, & 0\le i\le n-s-2,\\
s+1+r, & n-s-1\le i\le n-r-1,\\
r, & n-r\le i\le n-1 .
\end{cases}
\end{equation}

\noindent When $r=0$, all $L_i=s+1$ and $\mathbf{B}$ is the circulant matrix
\begin{equation}\label{MatrixB}
\mathbf{B}=\Circ\big(\one_{1\times(s+1)},\,\mathbf{0}_{1\times(n-s-1)}\big)\in\{0,1\}^{n\times n},
\end{equation}
where $\Circ(\mathbf{c})$ denotes the circulant matrix whose row $i$ is the right cyclic shift of $\mathbf{c}$ by $i$ positions, so that $\mathbf{b}_i$ is the cyclic shift of $\mathbf{b}_0=\one_{S_0}$ by $i$ positions. When $1\le r\le s$, the lengths $L_i$ are no longer constant and $\mathbf{B}$ splits into two horizontal blocks,
\begin{equation}\label{MatrixB-nd}
\mathbf{B}=\begin{pmatrix}\mathbf{B}_{1}\\[2pt] \mathbf{B}_{2}\end{pmatrix}\in\{0,1\}^{n\times n},
\end{equation}
where $\mathbf{B}_{1}\in\{0,1\}^{(n-s-1)\times n}$ collects the rows $\mathbf{b}_i$ with $0\le i\le n-s-2$, which are the first $n-s-1$ rows of the circulant matrix generated by $\mathbf{b}_0$, and $\mathbf{B}_{2}\in\{0,1\}^{(s+1)\times n}$ collects the rows of the last workers of the $s+1$ classes, of weight $s+1+r$ or $r$. We then call $\mathbf{B}$ \emph{near-circulant}, since it differs from $\Circ(\mathbf{b}_0)$ only in these $s+1$ rows, one for each class. Example~\ref{example1} displays $\mathbf{B}$ for $n=6$, $s=2$ (divisible), and Example~\ref{ex:nd} for $n=11$, $s=2$ (non-divisible).

\noindent The following result is the basis of our construction: for each class $K_j$, the supports $S_i$, $i\in K_j$, partition $[n]$.

\begin{proposition}\label{partition1}
For every $j\in\{0,\ldots,s\}$, the family $\{S_i\mid i\in K_j\}$ is a partition of $[n]$, i.e., $\sum_{i\in K_j}\mathbf{b}_i=\one_{1\times n}$.
\end{proposition}

\begin{proof}
By Definition~\ref{Def_cyclic-support}, each $S_{i_q}$ is the cyclic interval of length $L_{i_q}$ starting at $i_q$.
\begin{enumerate}
\item \emph{Consecutive intervals.} The starting points satisfy
\[
i_q+L_{i_q}=i_{q+1}\quad(0\le q\le m_j-2),
\qquad
i_{m_j-1}+L_{i_{m_j-1}}=j+n\equiv i_0\pmod n .
\]
Hence each interval begins where the previous one ends, and $S_{i_0},\ldots,S_{i_{m_j-1}}$ wrap once around $[n]$.
\item \emph{Cover.} Their lengths sum to $n$, so they cover $[n]$ exactly once, that is, $\{S_i\}_{i\in K_j}$ partitions $[n]$. Since the $\mathbf{b}_i$ are binary, this is equivalent to
\[
\sum_{i\in K_j}\mathbf{b}_i=\one_{1\times n}.
\]
\end{enumerate}
\end{proof}

\noindent We next turn to the storage cost of the scheme and show that it is the least possible: each chunk is replicated on exactly $s+1$ workers, the minimum needed to tolerate $s$ stragglers.

\begin{proposition}\label{prop:balance}
Let $\mathbf{B}$ be the encoding matrix of the BCGC scheme. Then every column of $\mathbf{B}$ has weight $s+1$, and the number of nonzero entries $\nnzr(\mathbf{B})=n(s+1)$ is the minimal total load for a scheme tolerating $s$ stragglers. Moreover, when $r=0$, every row also has weight $s+1$.
\end{proposition}

\begin{proof}
\begin{enumerate}
\item \emph{Column weights.} By Proposition~\ref{partition1}, the supports of each class partition $[n]$, so each chunk lies in exactly one support of each of the $s+1$ classes. Hence every column of $\mathbf{B}$ has weight $s+1$, and
\[
\nnzr(\mathbf{B})=n(s+1).
\]
\item \emph{Minimality.} To tolerate $s$ stragglers, each chunk must be stored on at least $s+1$ workers, so $\nnzr(\mathbf{B})\ge n(s+1)$ for every such scheme. Thus $n(s+1)$ is the minimal total load.
\item \emph{Row weights when $r=0$.} Here \eqref{eq:nd-Li} gives $L_i=s+1$ for every $i$, so every row of $\mathbf{B}$ also has weight $s+1$.
\end{enumerate}
\end{proof}

\noindent The following theorem shows that the full gradient is recovered by summing the responses of a set of workers if and only if that set is a single class $K_j$. This is a restriction compared with FRC and GFRC \cite{Tandon2017,Charalambides2025}, where workers of different classes may also be combined, so that those schemes have more binary recovering sets.

\begin{theorem}\label{thm:char}
Let $\mathbf{B}$ have rows $\mathbf{b}_i=\one_{S_i}$ as above, with $L_i=|S_i|$ given by \eqref{eq:nd-Li}.
\begin{enumerate}
\item For every $R\subseteq[n]$, the sum being taken over $\mathbb{R}$,
\begin{align*}
\sum_{i\in R}\mathbf{b}_i=\one_{1\times n}
&\iff\{S_i\}_{i\in R}\ \text{is a partition of}\ [n]\\
&\iff R=K_j\ \text{for some}\ j .
\end{align*}
\item If $I\subseteq[n]$ with $|I|\ge n-s$, then $K_{j^*}\subseteq I$ for some $j^*\in\{0,\ldots,s\}$, and, with $\mathbf{a}_I=\one_{K_{j^*}}$,
\[
\mathbf{a}_I\mathbf{B}=\one_{1\times n},\qquad \sum_{i\in K_{j^*}}\sum_{l\in S_i}\mathbf{g}_l=\mathbf{g} .
\]
The scheme tolerates any $s$ stragglers, and its recovery threshold equals $n-s$.
\end{enumerate}
\end{theorem}

\begin{proof}
Throughout, indices are taken modulo $n$ and $[n]$ is identified with $\mathbb Z_n$. For $R\subseteq[n]$ and $l\in[n]$, let
\[
n_R(l):=|\{i\in R\mid l\in S_i\}| ,
\]
the $l$-th coordinate of $\sum_{i\in R}\mathbf{b}_i$. Then
\[
\sum_{i\in R}\mathbf{b}_i=\one_{1\times n}
\iff n_R(l)=1\ \text{for all}\ l
\iff \{S_i\}_{i\in R}\ \text{partitions}\ [n],
\]
which is the first equivalence in part~(1).

\textit{Part (1): $\{S_i\}_{i\in R}$ partitions $[n]$ if and only if $R=K_j$ for some $j$.}
\begin{enumerate}
\item If $R=K_j$, this holds by Proposition~\ref{partition1}.
\item Conversely, assume $n_R(l)=1$ for all $l$ and fix $a\in R$. We consider two cases.

\emph{Case $L_a=n$.} Then $S_a=[n]$, so $R=\{a\}$. Moreover $L_a=n$ forces $\ell=1$ with $a$ the unique worker of its class, so $R$ is a class.

\emph{Case $L_a<n$.} Then $a+L_a\notin S_a$, so the unique $i\in R$ with $a+L_a\in S_i$ satisfies $i\ne a$. If $i\ne a+L_a$, then $a+L_a=i+t$ with $1\le t\le L_i-1$, hence $a+L_a-1\in S_i\cap S_a$ and $n_R(a+L_a-1)\ge2$, contradicting $n_R\equiv1$. Therefore $i=a+L_a$, which proves
\begin{equation}\label{eq:nd-succ}
a\in R\implies a+L_a\in R .
\end{equation}
Let $K_{j^*}=\{i_0,\ldots,i_{m_{j^*}-1}\}$ be the class of $a$. Since
\[
i_q+L_{i_q}=i_{q+1}\ (q\le m_{j^*}-2),
\qquad
i_{m_{j^*}-1}+L_{i_{m_{j^*}-1}}\equiv i_0\pmod n,
\]
iterating \eqref{eq:nd-succ} from $a$ yields $K_{j^*}\subseteq R$. The supports of $K_{j^*}$ already partition $[n]$, so any element of $R\setminus K_{j^*}$ would cover some index twice. Hence $R=K_{j^*}$.
\end{enumerate}

\textit{Part (2): straggler tolerance and threshold.}
\begin{enumerate}
\item \emph{Recovery from $n-s$ responses.} Let $T=[n]\setminus I$, so $|T|\le s$. The $s+1$ classes are disjoint, so $T$ meets at most $s$ of them, and by the pigeonhole principle $K_{j^*}\subseteq I$ for some $j^*$. By part~(1),
\[
\sum_{i\in K_{j^*}}\sum_{l\in S_i}\mathbf{g}_l=\sum_{l\in[n]}\mathbf{g}_l=\mathbf{g} .
\]
\item \emph{Optimality of the threshold.} The threshold cannot be lowered, even by real-valued or nonlinear decoding. By Proposition~\ref{prop:balance}, each chunk $l$ belongs to exactly $s+1$ supports. If these $s+1$ workers straggle, the $n-s-1$ received messages do not depend on $\mathbf{g}_l$, so $\mathbf{g}$ cannot be recovered. Hence the recovery threshold equals $n-s$.
\end{enumerate}
\end{proof}

\noindent\textbf{Decoding.} Since $\mathbf{B}$ is binary, no linear system needs to be solved: by Theorem~\ref{thm:char} it suffices to find a complete class and to sum the messages of its workers. Algorithm~\ref{alg:dec} processes the responses in their order of arrival, counts those received from each class, and stops as soon as one class is complete.

\begin{algorithm}[H]
\caption{Decoding by complete class}\label{alg:dec}
\begin{algorithmic}[1]
\State \textbf{Input} parameters $n$, $s$ with $n=\ell(s+1)+r$, and the messages $\widehat{\mathbf{g}}_{i_1},\widehat{\mathbf{g}}_{i_2},\ldots$ of the responding workers, in their order of arrival
\State \textbf{Output} a complete class $K_{j^*}$, used to recover $\mathbf{g}=\sum_{i\in K_{j^*}}\widehat{\mathbf{g}}_i$
\State $c_j\gets0$, and $m_j\gets\ell+1$ if $j<r$, $m_j\gets\ell$ otherwise, for $j=0,\ldots,s$ \Comment{$m_j=|K_j|$}
\For{$t=1$ \textbf{to} $n-s$}
  \State $j\gets i_t\bmod(s+1)$; \quad $c_j\gets c_j+1$ \Comment{the class of worker $i_t$}
  \If{$c_j=m_j$}
    \State \Return $K_j=\{\,j+q(s+1)\mid q=0,\ldots,m_j-1\,\}$
  \EndIf
\EndFor
\end{algorithmic}
\end{algorithm}

\noindent\textbf{Correctness and complexity.} The workers $i_t$ are distinct and $|K_j|=m_j$, so $c_j=m_j$ exactly when all workers of $K_j$ have responded, and line~7 returns the first class to complete. After $n-s$ responses, at most $s$ workers are missing, meeting at most $s$ of the $s+1$ disjoint classes, so some class is complete and the loop always returns (Theorem~\ref{thm:char}(2)), with $\sum_{i\in K_{j^*}}\widehat{\mathbf{g}}_i=\mathbf{g}$ by Theorem~\ref{thm:char}(1). Each response costs $O(1)$, and the final sum of $m_{j^*}\le\ell+1$ messages of dimension $d$ costs $O((\ell+1)d)$, with no decoding matrix stored. This decoder is used in the experiments of Section~\ref{Sec_Experiments}.

\begin{example}[Divisible case, $n=6$, $s=2$]\label{example1}
Here $\ell=2$, $r=0$, $K_0=\{0,3\}$, $K_1=\{1,4\}$, $K_2=\{2,5\}$, and $S_i=\{i,i+1,i+2\}\pmod 6$. Then

$
\mathbf{B}=\Circ(1,1,1,0,0,0)=
\begin{pmatrix}
1&1&1&0&0&0\\ 0&1&1&1&0&0\\ 0&0&1&1&1&0\\
0&0&0&1&1&1\\ 1&0&0&0&1&1\\ 1&1&0&0&0&1
\end{pmatrix},
$
and each class covers $[6]$ exactly once, $\mathbf{b}_0+\mathbf{b}_3=\mathbf{b}_1+\mathbf{b}_4=\mathbf{b}_2+\mathbf{b}_5=\one_{1\times 6}$, yielding the three distinct data partitions $\{\{0,1,2\},\{3,4,5\}\}$, $\{\{1,2,3\},\{4,5,0\}\}$ and $\{\{2,3,4\},\{5,0,1\}\}$. If workers $0$ and $1$ straggle, the class $K_2=\{2,5\}$ is complete and the master recovers $(\mathbf{g}_2+\mathbf{g}_3+\mathbf{g}_4)+(\mathbf{g}_5+\mathbf{g}_0+\mathbf{g}_1)=\mathbf{g}$; if instead workers $1$ and $2$ straggle, the complete class is $K_0=\{0,3\}$, giving $(\mathbf{g}_0+\mathbf{g}_1+\mathbf{g}_2)+(\mathbf{g}_3+\mathbf{g}_4+\mathbf{g}_5)=\mathbf{g}$.
\end{example}

\begin{example}[Non-divisible case, $n=11$, $s=2$]\label{ex:nd}
Here $\ell=3$, $r=2$, and $K_0=\{0,3,6,9\}$, $K_1=\{1,4,7,10\}$, $K_2=\{2,5,8\}$. The first workers of each class receive $s+1=3$ consecutive indices and the last receives the remaining ones; for $K_2$, the supports are $S_2=\{2,3,4\}$, $S_5=\{5,6,7\}$ and $S_8=\{8,9,10,0,1\}$. Hence
\[
\mathbf{B} =
\left(
\begin{array}{ccccccccccc}
1&1&1&0&0&0&0&0&0&0&0\\
0&1&1&1&0&0&0&0&0&0&0\\
0&0&1&1&1&0&0&0&0&0&0\\
0&0&0&1&1&1&0&0&0&0&0\\
0&0&0&0&1&1&1&0&0&0&0\\
0&0&0&0&0&1&1&1&0&0&0\\
0&0&0&0&0&0&1&1&1&0&0\\
0&0&0&0&0&0&0&1&1&1&0\\
1&1&0&0&0&0&0&0&1&1&1\\
0&0&0&0&0&0&0&0&0&1&1\\
1&0&0&0&0&0&0&0&0&0&1
\end{array}
\right),
\]
whose first $n-s-1=8$ rows are consecutive cyclic shifts of $\mathbf{b}_0=\one_{\{0,1,2\}}$ and form the block $\mathbf{B}_1$, while the last three rows, the block $\mathbf{B}_2$ of the workers $8,9,10$, have weights $5$, $2$ and $2$. Every column has weight $3$, and the rows of each class sum to $\one_{1\times 11}$; for instance $\mathbf{b}_2+\mathbf{b}_5+\mathbf{b}_8=\one_{1\times 11}$. If workers $0$ and $1$ straggle, the class $K_2=\{2,5,8\}$ is complete and the master recovers $\mathbf{g}=(\mathbf{g}_2+\mathbf{g}_3+\mathbf{g}_4)+(\mathbf{g}_5+\mathbf{g}_6+\mathbf{g}_7)+(\mathbf{g}_8+\mathbf{g}_9+\mathbf{g}_{10}+\mathbf{g}_0+\mathbf{g}_1)$. The threshold $n-s=9$ cannot be lowered: if workers $0$, $1$ and $2$ straggle, chunk $2$ is held only by them, so no decoder, even real-valued or nonlinear, recovers $\mathbf{g}$.
\end{example}

\begin{remark}\label{rem:cooperation-div}
By Theorem~\ref{thm:char}(1), binary recovery in BCGC requires all the workers of a single class. In $\mathbf{B}'_{\mathrm{FRC}}=I_{\ell}\otimes\one_{(s+1)\times(s+1)}$, by contrast, the workers of each diagonal block share the same support, so one worker per block suffices, even from incomplete classes. Thus FRC has $(s+1)^\ell$ binary recovering sets, against $s+1$ for BCGC, and decodes from more straggler patterns; in return, the $s+1$ classes of BCGC induce $s+1$ distinct data partitions, whereas all the recovering sets of FRC induce the same one and hence return the same compressed update (Proposition~\ref{prop:diversity}).
\end{remark}

\subsection{Algebraic structure and comparison with FRC and GFRC}\label{Sec_comp-frc}

We describe the code $\mathcal C_B\subseteq\mathbb R^{n}$ spanned by the rows of $\mathbf{B}$ over $\mathbb R$, the field in which decoding takes place, and compare $\mathbf{B}$ with the encoding matrices of FRC \cite{Tandon2017} and GFRC \cite{Charalambides2025}. We identify $\mathbf{c}=(c_0,\ldots,c_{n-1})\in\mathbb R^{n}$ with the polynomial $c(x)=\sum_{j}c_jx^j$ in $\mathbb R[x]/(x^n-1)$, and write $\mathbf{e}_l:=\one_{\{l\}}$. Two matrices are \emph{permutation-equivalent} if one is obtained from the other by permuting rows and columns; such permutations preserve the rank, so matrices of different ranks are not permutation-equivalent.

\noindent We first determine the dimension of $\mathcal C_B$, its cyclic structure in the divisible case, and whether it is MDS; the rank will then separate $\mathbf{B}$ from the FRC and GFRC matrices.

\begin{proposition}\label{prop:code}
We have $\dim\mathcal C_B=\operatorname{rank}\mathbf{B}=n-s$ and $\one_{1\times n}\in\mathcal C_B$. When $r=0$, $\mathcal C_B$ is the cyclic code with generator polynomial $b_0(x)=1+x+\cdots+x^{s}$. Moreover, $\mathcal C_B$ is not MDS, except when $r=0$ and $s\in\{0,1,n-1\}$.
\end{proposition}

\begin{proof}
First we start  by computing the rank of $B.$ For $0\le i\le n-s-2$, \eqref{eq:nd-Li} gives $L_i=s+1$, so $\mathbf{b}_i=\one_{\{i,\ldots,i+s\}}$ with $i+s\le n-2$. These $n-s-1$ rows are in echelon form, hence linearly independent, and each has last coordinate $0$. By Proposition~\ref{partition1},
\[
\one_{1\times n}=\sum_{i\in K_0}\mathbf{b}_i\in\mathcal C_B ,
\]
and its last coordinate is $1$, so adjoining it to the $n-s-1$ rows gives $n-s$ independent vectors and $\operatorname{rank}\mathbf{B}\ge n-s$. For the reverse bound, set
\[
\mathbf{v}_j:=\one_{K_0}-\one_{K_j},\qquad j=1,\ldots,s .
\]
By Proposition~\ref{partition1}, $\mathbf{v}_j\mathbf{B}=\mathbf 0$, and the $\mathbf{v}_j$ are independent because the classes are disjoint and nonempty. Thus the left kernel of $\mathbf{B}$ has dimension at least $s$, so $\operatorname{rank}\mathbf{B}\le n-s$. The two bounds give $\operatorname{rank}\mathbf{B}=n-s$.

Next we prove that the code $ \mathcal C_B$ is cyclic when $r=0$. Each row $\mathbf{b}_i$ is the cyclic shift of $\mathbf{b}_0$ by $i$ positions, so $b_i(x)=x^ib_0(x)$ in $\mathbb R[x]/(x^n-1)$, and $\mathcal C_B$ is the cyclic code generated by $b_0(x)$. From
\[
x^n-1=(x-1)\,b_0(x)\sum_{l=0}^{\ell-1}x^{l(s+1)},
\]
$b_0(x)$ is a monic divisor of $x^n-1$, hence its generator polynomial, and $\dim\mathcal C_B=n-\deg b_0=n-s$.
Finally we prove that  $ \mathcal C_B$ is not MDS.  By the Singleton bound, $d_{\min}(\mathcal C_B)\le n-\dim\mathcal C_B+1=s+1$, with equality if and only if $\mathcal C_B$ is MDS. It suffices to exhibit a codeword of weight less than $s+1$:
\begin{itemize}
\item if $r=1$, then $L_{n-1}=1$, so $\mathbf{b}_{n-1}=\mathbf{e}_{n-1}$ has weight $1$;
\item if $r\ge2$ (so $s\ge2$ and $n-s-1\ge2$), or if $r=0$ and $2\le s\le n-2$, then $\mathbf{b}_0-\mathbf{b}_1=\mathbf{e}_0-\mathbf{e}_{s+1}$ has weight $2<s+1$.
\end{itemize}
In the remaining cases $r=0$ and $s\in\{0,1,n-1\}$, the code is the trivial $[n,n,1]$ code ($s=0$), the $[n,n-1,2]$ code ($s=1$), or the $[n,1,n]$ repetition code ($s=n-1$), all of which are MDS.
\end{proof}

\noindent We now recall the encoding of GFRC \cite{Charalambides2025}, which extends FRC to the non-divisible case. It uses one partition of $[n]$ for the $r$ classes of size $\ell+1$ and another for the $s+1-r$ classes of size $\ell$.

\begin{definition}[GFRC encoding matrix]\label{def:gfrc}
Write $n=\lambda(\ell+1)+\tilde r$ with $0\le\tilde r\le\ell$, and $n=\mu\ell+u$ with $0\le u\le\ell-1$. For $j\in\{0,\ldots,s\}$ and $i_q=j+q(s+1)\in K_j$, let
\[
L'_{i_q}:=
\begin{cases}
\lambda+1, & j<r,\ q<\tilde r,\\
\lambda, & j<r,\ q\ge\tilde r,\\
\mu+1, & j\ge r,\ q<u,\\
\mu, & j\ge r,\ q\ge u,
\end{cases}
\qquad
S'_{i_q}:=\{\alpha'_{i_q},\ldots,\alpha'_{i_q}+L'_{i_q}-1\},
\]
where $\alpha'_{i_q}:=\sum_{l=0}^{q-1}L'_{i_{l}}$, so that the first worker $i_0=j$ of each class receives a support starting at chunk $0$.
\end{definition}

\noindent The lengths of each class sum to $n$, so the supports $S'_i$, $i\in K_j$, are consecutive blocks partitioning $[n]$, depending on $j$ only through the size of $K_j$. All classes of size $\ell+1$ thus use one partition $\mathcal G_1$ of $[n]$, and all classes of size $\ell$ use another partition $\mathcal G_2$. Unlike the cyclic supports $S_i$ of BCGC, the supports $S'_i$ are not cyclic.

\noindent Comparing the ranks of the three encoding matrices now shows that BCGC is a genuinely new construction, not a relabelling of FRC or GFRC.

\begin{theorem}[Non-equivalence with FRC and GFRC]\label{thm:not-frc}
We have $\operatorname{rank}\mathbf{B}=n-s$, whereas $\operatorname{rank}\mathbf{B}'_{\mathrm{FRC}}=\ell$ when $r=0$ and $\operatorname{rank}\mathbf{B}_{\mathrm{GFRC}}\le2\ell$ when $1\le r\le s$. Consequently, $\mathbf{B}$ is not permutation-equivalent to $\mathbf{B}'_{\mathrm{FRC}}$ for $1\le s\le n-2$, nor to $\mathbf{B}_{\mathrm{GFRC}}$ when $s\ge2$ and $n\ne s+2$.
\end{theorem}

\begin{proof}
By Proposition~\ref{prop:code}, $\operatorname{rank}\mathbf{B}=n-s$. We compute the other two ranks and compare.
\begin{enumerate}
\item \emph{FRC ($r=0$).} The matrix $\mathbf{B}'_{\mathrm{FRC}}=I_{\ell}\otimes\one_{(s+1)\times(s+1)}$ is block diagonal with $\ell$ rank-$1$ blocks, so $\operatorname{rank}\mathbf{B}'_{\mathrm{FRC}}=\ell$. For $1\le s\le n-2$ we have $\ell\ge2$, hence
\[
(n-s)-\ell=s(\ell-1)>0 .
\]
\item \emph{GFRC ($1\le r\le s$).} Definition~\ref{def:gfrc} gives $\mathbf{b}'_{j+q(s+1)}=\mathbf{b}'_{q(s+1)}$ for $j<r$ and $\mathbf{b}'_{j+q(s+1)}=\mathbf{b}'_{s+q(s+1)}$ for $j\ge r$. Hence the row space is $V_0+V_s$, where
\[
V_0:=\operatorname{span}\{\mathbf{b}'_i\mid i\in K_0\},\qquad
V_s:=\operatorname{span}\{\mathbf{b}'_i\mid i\in K_s\}
\]
have dimensions at most $\ell+1$ and $\ell$. Since $\one_{1\times n}\in V_0\cap V_s$,
\[
\begin{aligned}
\operatorname{rank}\mathbf{B}_{\mathrm{GFRC}}
&=\dim V_0+\dim V_s-\dim(V_0\cap V_s)\\
&\le(\ell+1)+\ell-1=2\ell .
\end{aligned}
\]
Moreover $(n-s)-2\ell=(\ell-1)(s-1)+(r-1)\ge0$, and this is positive unless $\ell=r=1$, i.e.\ $n=s+2$.
\end{enumerate}
In every case covered by the statement the ranks differ, so the matrices are not permutation-equivalent.
\end{proof}

\noindent Beyond the rank, the two schemes also differ in how they spread the storage load across the workers. The next remark records the loads of BCGC in both the divisible and non-divisible cases, where it attains the same minimal total deviation as GFRC but over a wider range.

\begin{remark}[Load distribution]\label{rem:loads}
The storage load of BCGC depends on whether $(s+1)\mid n$. We measure its imbalance by the total deviation $d_s(\mathbf{B}):=\sum_{i\in[n]}\big|\|\mathbf{b}_i\|_0-(s+1)\big|$ \cite[Def.~6]{Charalambides2025}.
\begin{enumerate}
\item \emph{Divisible case ($r=0$).} By \eqref{eq:nd-Li} every row of $\mathbf{B}$ has weight $s+1$, so all workers carry the same load $s+1$ and $d_s(\mathbf{B})=0$.
\item \emph{Non-divisible case ($1\le r\le s$).} The loads cannot all equal $s+1$. By \eqref{eq:nd-Li} only the last worker of each class deviates, with load $r$ or $s+1+r$, so $d_s(\mathbf{B})=2r(s+1-r)=d_s(\mathbf{B}_{\mathrm{GFRC}})$ is minimal \cite[Th.~11]{Charalambides2025}. The loads of $\mathbf{B}$ range over $[r,s+1+r]$, wider than $[\lfloor n/(\ell+1)\rfloor,\lceil n/\ell\rceil]$ for GFRC \cite[Th.~8]{Charalambides2025}. This is the price of the cyclic supports, which keep the maximum number of workers, $n-s-1$, at load $s+1$ and $\mathbf{B}$ as close as possible to circulant.
\end{enumerate}
\end{remark}

\section{Compatibility with gradient compression}\label{Sec_Compression}
Without compression, every complete class returns $\mathbf{g}$ exactly, so the decoded class does not matter. Under a nonlinear compressor, it does. In this section, we show that a compressed binary scheme is an uncoded compressed scheme on a coarser partition of the data, that the standard guarantee of compressed gradient descent with a contractive compressor holds on every such partition, and that the classes of BCGC induce more distinct partitions than those of FRC and GFRC.

Compressing the coded messages of gradient coding has been studied in \cite{Li2024OneBit,Li2026Biased,LiCyclic2026,SongChoi2026}, with real-valued or robust-aggregation decoding. 

Each worker $i$ sends $\mathcal{C}(\widehat{\mathbf{g}}_i)$, where $\mathcal{C}:\mathbb{R}^{1\times d}\to\mathbb{R}^{1\times d}$ denotes compression followed by decompression. Given a decoding vector $\mathbf{a}_I$ with $\operatorname{supp}(\mathbf{a}_I)\subseteq I$ and $\mathbf{a}_I\mathbf{B}=\one_{1\times n}$, the master forms
\begin{equation}\label{eq:coded}
\mathcal{C}_I^{\mathrm{cod}}(\mathbf{g}):=\sum_{i\in I}\mathbf{a}_I(i)\,\mathcal{C}(\widehat{\mathbf{g}}_i).
\end{equation}

\begin{remark}[Linear compressors]\label{rem:linear}
If $\mathcal{C}$ is linear, then $\mathcal{C}^{\mathrm{cod}}_I(\mathbf{g})=\mathcal{C}(\mathbf{a}_I\mathbf{B}\bar{\mathbf{g}})=\mathcal{C}(\one_{1\times n}\bar{\mathbf{g}})=\mathcal{C}(\mathbf{g})$, for every scheme and every decoding vector.
\end{remark}

\noindent Most compressors used in practice are nonlinear, such as $\mathrm{Top}\text{-}\kappa$ sparsification \cite{Alistarh2018Topk,Lin2018DGC}, quantization \cite{Alistarh2017QSGD} and sign compression \cite{Karimireddy2019EF}. From now on, $\mathcal{C}$ is deterministic and possibly nonlinear, and we write $\boldsymbol{\varepsilon}_i:=\mathcal{C}(\widehat{\mathbf{g}}_i)-\widehat{\mathbf{g}}_i$ for the compression error of worker $i$.

\noindent We first bound the error of the decoded update by the decoding coefficients times the per-worker compression errors. This makes explicit why real-valued schemes with large coefficients are unstable under compression.

\begin{proposition}[Error of the decoded update]\label{prop:diff}
Let $\mathbf{a}_I$ be a decoding vector of an encoding matrix $\mathbf{B}$. Then, for every compressor $\mathcal{C}$,
\[
\mathcal{C}^{\mathrm{cod}}_I(\mathbf{g})-\mathbf{g}=\sum_{i\in I}\mathbf{a}_I(i)\,\boldsymbol{\varepsilon}_i ,
\]
and
\begin{equation*}
\begin{split}
\big\|\mathcal{C}^{\mathrm{cod}}_I(\mathbf{g})-\mathbf{g}\big\|
\le\min\Big\{&\|\mathbf{a}_I\|_\infty\!\!\sum_{i\in\operatorname{supp}(\mathbf{a}_I)}\!\!\|\boldsymbol{\varepsilon}_i\|,\\
&\|\mathbf{a}_I\|_1\max_{i\in\operatorname{supp}(\mathbf{a}_I)}\|\boldsymbol{\varepsilon}_i\|\Big\}.
\end{split}
\end{equation*}
\end{proposition}

\begin{proof}
Substituting $\mathcal{C}(\widehat{\mathbf{g}}_i)=\widehat{\mathbf{g}}_i+\boldsymbol{\varepsilon}_i$ into \eqref{eq:coded} and using $\sum_{i\in I}\mathbf{a}_I(i)\,\widehat{\mathbf{g}}_i=\mathbf{a}_I\mathbf{B}\bar{\mathbf{g}}=\mathbf{g}$ gives the identity. Restricting to $\operatorname{supp}(\mathbf{a}_I)$, since the other terms vanish, the triangle inequality gives
$$
\big\|\mathcal{C}^{\mathrm{cod}}_I(\mathbf{g})-\mathbf{g}\big\|
=\Big\|\sum_{i\in\operatorname{supp}(\mathbf{a}_I)}\mathbf{a}_I(i)\,\boldsymbol{\varepsilon}_i\Big\|
\le\sum_{i\in\operatorname{supp}(\mathbf{a}_I)}|\mathbf{a}_I(i)|\,\|\boldsymbol{\varepsilon}_i\|,
$$
and bounding $|\mathbf{a}_I(i)|$ by $\|\mathbf{a}_I\|_\infty$, or the sum $\sum_i|\mathbf{a}_I(i)|=\|\mathbf{a}_I\|_1$ times the maximum error, gives the two terms of the minimum.
\end{proof}

\noindent Large decoding coefficients, as in real-valued schemes, thus amplify the compression errors. Binary decoding avoids this.

\begin{corollary}[Binary decoding]\label{cor:binary-bound}
Let $\mathbf{B}\in\{0,1\}^{n\times n}$, and let $R\subseteq I$ satisfy $\sum_{i\in R}\mathbf{b}_i=\one_{1\times n}$. Then $\one_R$ is a decoding vector, and, decoding with $\mathbf{a}_I=\one_R$,
\begin{equation}\label{eq:binary-cod}
\begin{gathered}
\mathcal{C}^{\mathrm{cod}}_I(\mathbf{g})=\sum_{i\in R}\mathcal{C}(\widehat{\mathbf{g}}_i),\\
\big\|\mathcal{C}^{\mathrm{cod}}_I(\mathbf{g})-\mathbf{g}\big\|\le\sum_{i\in R}\|\boldsymbol{\varepsilon}_i\|\le|R|\max_{i\in R}\|\boldsymbol{\varepsilon}_i\| .
\end{gathered}
\end{equation}
\end{corollary}

\begin{proof}
We have $\operatorname{supp}(\one_R)=R\subseteq I$ and $\one_R\mathbf{B}=\one_{1\times n}$, so $\one_R$ is a decoding vector. Since $\|\one_R\|_\infty=1$ and $\|\one_R\|_1=|R|$, Proposition~\ref{prop:diff} gives the two bounds.
\end{proof}

The next result is the key to the analysis. Binary decoding turns the compressed coded update into an uncoded compressed update on a coarser partition of the data.

\begin{proposition}[Coarser partition]\label{prop:coarsening}
Assume that the loss is a sum of per-sample losses, as in Section~\ref{Sec_Prelim}. Let $\mathbf{B}\in\{0,1\}^{n\times n}$ have nonzero rows, let $R$ be as in Corollary~\ref{cor:binary-bound}, and, for $i\in R$, let $S_i:=\operatorname{supp}(\mathbf{b}_i)$ and $\bar{\mathcal{D}}_i:=\bigcup_{j\in S_i}\mathcal{D}_j$. Then $\mathcal{D}=\bigsqcup_{i\in R}\bar{\mathcal{D}}_i$, and, decoding with $\one_R$,
\begin{equation}\label{eq:coarsening}
\mathcal{C}^{\mathrm{cod}}_I(\mathbf{g})=\sum_{i\in R}\mathcal{C}\big(\nabla L(\bar{\mathcal{D}}_i,\boldsymbol{\beta})\big).
\end{equation}
That is, the decoded update is the naive compressed update on the coarse partition $\{\bar{\mathcal{D}}_i\}_{i\in R}$ instead of the fine partition $\{\mathcal{D}_j\}_{j\in[n]}$.
\end{proposition}

\begin{proof}
Since the rows are binary, nonzero, and $\sum_{i\in R}\mathbf{b}_i=\one_{1\times n}$, we have $[n]=\bigsqcup_{i\in R}S_i$, hence $\mathcal{D}=\bigsqcup_{i\in R}\bar{\mathcal{D}}_i$. Since the loss is a sum of per-sample losses and the chunks are disjoint, $$\widehat{\mathbf{g}}_i=\sum_{j\in S_i}\nabla L(\mathcal{D}_j,\boldsymbol{\beta})=\nabla L(\bar{\mathcal{D}}_i,\boldsymbol{\beta}).$$ Substituting into \eqref{eq:binary-cod} gives \eqref{eq:coarsening}.
\end{proof}

\noindent Since $\mathcal{C}$ is nonlinear, compressing the block gradients $\widehat{\mathbf{g}}_i$ differs from compressing the chunk gradients $\mathbf{g}_j$. We now show that the standard convergence guarantee of compressed gradient descent with a contractive compressor \cite{Stich2018,Karimireddy2019EF}, which holds on the fine partition, also holds on every coarse partition.

\begin{definition}[Contractive compressor \cite{Stich2018,Karimireddy2019EF}]\label{def:contractive}
A compressor $\mathcal C$ is \emph{$\delta$-contractive}, with $\delta\in(0,1]$, if $\|\mathcal C(\mathbf x)-\mathbf x\|^2\le(1-\delta)\|\mathbf x\|^2$ for all $\mathbf x\in\mathbb R^{1\times d}$.
\end{definition}

To prove convergence on a coarse partition, we use standard quantities from federated optimization with unbalanced local datasets \cite{Li2020FedProx,Wang2020FedNova}, weighting each chunk by its fraction of the data and measuring heterogeneity by the weighted variance of the gradients. For $i\in[n]$, let $S_i$ and $\bar{\mathcal D}_i$ be as in Proposition~\ref{prop:coarsening}, set $w_j:=|\mathcal D_j|/N$ and $\bar w_i:=|\bar{\mathcal D}_i|/N=\sum_{j\in S_i}w_j$, and define
\begin{align*}
\zeta^2(\boldsymbol\beta)
&:=\sum_{j\in[n]}
\frac{||\mathbf g_j-w_j\mathbf g||^2}{w_j},
&\bar\zeta_R^2(\boldsymbol\beta):=\sum_{i\in R}
\frac{||\widehat{\mathbf g}_i-\bar w_i\mathbf g||^2}{\bar w_i},
\end{align*}
the heterogeneity of the fine and coarse partitions. Thus $\zeta^2(\boldsymbol\beta)=0$ iff $\mathbf g_j=w_j\mathbf g$ for all $j$, and $\bar\zeta_{[n]}^2=\zeta^2$ when $\mathbf B=I_n$. At round $t=0,\ldots,T-1$, let $I_t$ be the responding set and let $\mathbf g^{(t)}$, $\widehat{\mathbf g}_i^{(t)}$ be the quantities evaluated at $\boldsymbol\beta_t$. We state the following standard assumptions.
\begin{assumption}\label{ass:transfer}
\begin{enumerate}
\item The loss is additive over samples, that is, $L(\mathcal A\sqcup\mathcal A',\cdot)=L(\mathcal A,\cdot)+L(\mathcal A',\cdot)$ for disjoint $\mathcal A,\mathcal A'\subseteq\mathcal D$.
\item The function $L(\mathcal D,\cdot)$ is differentiable, $\rho$-smooth and bounded below by $L^*\in\mathbb R$, that is, for all $\boldsymbol\beta,\boldsymbol\beta'\in\mathbb R^{d}$,
$$
\begin{gathered}
\big\|\nabla L(\mathcal D,\boldsymbol\beta)-\nabla L(\mathcal D,\boldsymbol\beta')\big\|\le\rho\,\|\boldsymbol\beta-\boldsymbol\beta'\|\\
\text{and}\quad L(\mathcal D,\boldsymbol\beta)\ge L^* .
\end{gathered}
$$
\item The compressor $\mathcal C$ is deterministic and $\delta$-contractive for some $\delta\in(0,1]$.
\item There exist $\zeta_*\ge0$ and $M\ge0$ with $(1-\delta)(1+M)<1$ such that $\zeta^2(\boldsymbol\beta)\le\zeta_*^2+M\,\|\nabla L(\mathcal D,\boldsymbol\beta)\|^2$ for all $\boldsymbol\beta$.
\end{enumerate}
\end{assumption}
\noindent The following proposition establishes the convergence of compressed gradient descent on coarse partitions, and we then show that FRC, GFRC and BCGC satisfy its conditions.
\begin{proposition}[Convergence on coarse partitions]\label{prop:transfer}
Let items 1 to 3 of Assumption~\ref{ass:transfer} hold, let $0<\gamma\le1/\rho$, and let $\mathbf B\in\{0,1\}^{n\times n}$ have nonzero rows. At each round $t$, let the master decode with $\one_{R_t}$, where $R_t\subseteq I_t$ and $\sum_{i\in R_t}\mathbf b_i=\one_{1\times n}$, and let $\boldsymbol\beta_{t+1}=\boldsymbol\beta_t-\gamma\,\mathcal C^{\mathrm{cod}}_{I_t}(\mathbf g^{(t)})$.
\begin{enumerate}
\item The iterates satisfy
\[
\frac1T\sum_{t=0}^{T-1}\|\mathbf g^{(t)}\|^2
\le\frac{2\big(L(\mathcal D,\boldsymbol\beta_0)-L^*\big)}{\delta\gamma T}
+\frac{1-\delta}{\delta\,T}\sum_{t=0}^{T-1}\bar\zeta_{R_t}^2(\boldsymbol\beta_t).
\]
\item If, moreover, item 4 holds, then, with $\delta_M:=1-(1-\delta)(1+M)>0$,
\[
\frac1T\sum_{t=0}^{T-1}\|\mathbf g^{(t)}\|^2
\le\frac{2\big(L(\mathcal D,\boldsymbol\beta_0)-L^*\big)}{\delta_M\gamma T}
+\frac{1-\delta}{\delta_M}\,\zeta_*^2 .
\]
\end{enumerate}
\end{proposition}

\begin{proof}
Fix $t$, and write $\mathbf g=\mathbf g^{(t)}$, $\widehat{\mathbf g}_i=\widehat{\mathbf g}_i^{(t)}.$ By item 1 of Assumption~\ref{ass:transfer} and Proposition~\ref{prop:coarsening}, $\sum_{i\in R_t}\widehat{\mathbf g}_i=\mathbf g$ and $\sum_{i\in R_t}\bar w_i=1$. Let $\mathbf u_i:=\widehat{\mathbf g}_i-\bar w_i\,\mathbf g$ for $i\in R_t$, so that $\sum_{i\in R_t}\mathbf u_i=\mathbf 0$.
\begin{enumerate}
\item Expanding $\|\bar w_i\mathbf g+\mathbf u_i\|^2$ and using $\sum_{i\in R_t}\bar w_i=1$ and $\sum_{i\in R_t}\mathbf u_i=\mathbf 0$ gives
\begin{equation}\label{eq:decomp}
\begin{split}
\sum_{i\in R_t}\frac{\|\widehat{\mathbf g}_i\|^2}{\bar w_i}
&=\sum_{i\in R_t}\Big(\bar w_i\|\mathbf g\|^2+2\langle\mathbf g,\mathbf u_i\rangle+\frac{\|\mathbf u_i\|^2}{\bar w_i}\Big)\\
&=\|\mathbf g\|^2+\bar\zeta_{R_t}^2(\boldsymbol\beta_t).
\end{split}
\end{equation}
By the convexity of $\|\cdot\|^2$, item 3 of Assumption~\ref{ass:transfer} and \eqref{eq:decomp},
\begin{equation}\label{eq:comp-err}
\begin{split}
\|\mathcal C^{\mathrm{cod}}_{I_t}(\mathbf g) -\mathbf g\|^2
&=\Big\|\sum_{i\in R_t}\bar w_i\,\frac{\mathcal C(\widehat{\mathbf g}_i)-\widehat{\mathbf g}_i}{\bar w_i}\Big\|^2\\
&\le\sum_{i\in R_t}\frac{\|\mathcal C(\widehat{\mathbf g}_i)-\widehat{\mathbf g}_i\|^2}{\bar w_i}
\le(1-\delta)\sum_{i\in R_t}\frac{\|\widehat{\mathbf g}_i\|^2}{\bar w_i}\\
&=(1-\delta)\big(\|\mathbf g\|^2+\bar\zeta_{R_t}^2(\boldsymbol\beta_t)\big).
\end{split}
\end{equation}
By item 2 of Assumption~\ref{ass:transfer}, the descent lemma with $\boldsymbol\beta_{t+1}=\boldsymbol\beta_t-\gamma\mathcal C^{\mathrm{cod}}_{I_t}(\mathbf g) $ gives
\[
\begin{split}
L(\mathcal D,\boldsymbol\beta_{t+1})\le L(\mathcal D,\boldsymbol\beta_t)
&-\gamma\langle\mathbf g,\mathcal C^{\mathrm{cod}}_{I_t}(\mathbf g) \rangle\\
&+\frac{\rho\gamma^2}{2}\|\mathcal C^{\mathrm{cod}}_{I_t}(\mathbf g) \|^2 .
\end{split}
\]
Since $\rho\gamma\le1$ and $-2\langle\mathbf g,\mathcal C^{\mathrm{cod}}_{I_t}(\mathbf g) \rangle+\|\mathcal C^{\mathrm{cod}}_{I_t}(\mathbf g) \|^2=\|\mathcal C^{\mathrm{cod}}_{I_t}(\mathbf g) -\mathbf g\|^2-\|\mathbf g\|^2$, we obtain, using \eqref{eq:comp-err},
\begin{align}
L(\mathcal D,\boldsymbol\beta_{t+1})
&\le L(\mathcal D,\boldsymbol\beta_t)-\frac{\gamma}{2}\big(\|\mathbf g\|^2-\|\mathcal C^{\mathrm{cod}}_{I_t}(\mathbf g) -\mathbf g\|^2\big)\nonumber\\
&\le L(\mathcal D,\boldsymbol\beta_t)-\frac{\gamma\delta}{2}\|\mathbf g\|^2+\frac{\gamma(1-\delta)}{2}\,\bar\zeta_{R_t}^2(\boldsymbol\beta_t).\label{eq:one-step}
\end{align}
Summing \eqref{eq:one-step} over $t=0,\ldots,T-1$ gives
\[
\frac{\gamma\delta}{2}\sum_{t=0}^{T-1}\|\mathbf g^{(t)}\|^2
\le L(\mathcal D,\boldsymbol\beta_0)-L(\mathcal D,\boldsymbol\beta_T)+\frac{\gamma(1-\delta)}{2}\sum_{t=0}^{T-1}\bar\zeta_{R_t}^2(\boldsymbol\beta_t).
\]
Using $L(\mathcal D,\boldsymbol\beta_T)\ge L^*$ and dividing by $\gamma\delta T/2$ gives item 1.

\item Since $\widehat{\mathbf g}_i=\sum_{j\in S_i}\mathbf g_j$ and $\bar w_i=\sum_{j\in S_i}w_j$, we have $\mathbf u_i=\sum_{j\in S_i}(\mathbf g_j-w_j\mathbf g)$. The Cauchy--Schwarz inequality gives
\[
\begin{split}
\|\mathbf u_i\|^2
&\le\Big(\sum_{j\in S_i}w_j\Big)\sum_{j\in S_i}\frac{\|\mathbf g_j-w_j\mathbf g\|^2}{w_j}\\
&=\bar w_i\sum_{j\in S_i}\frac{\|\mathbf g_j-w_j\mathbf g\|^2}{w_j}.
\end{split}
\]
Dividing by $\bar w_i$, summing over $i\in R_t$ and using $[n]=\bigsqcup_{i\in R_t}S_i$, we obtain, by item 4 of Assumption~\ref{ass:transfer},
\[
\bar\zeta_{R_t}^2(\boldsymbol\beta_t)\le\zeta^2(\boldsymbol\beta_t)\le\zeta_*^2+M\|\mathbf g\|^2 .
\]
Since $\delta-(1-\delta)M=1-(1-\delta)(1+M)=\delta_M$, inequality \eqref{eq:one-step} becomes
\[
L(\mathcal D,\boldsymbol\beta_{t+1})\le L(\mathcal D,\boldsymbol\beta_t)-\frac{\gamma\delta_M}{2}\|\mathbf g\|^2+\frac{\gamma(1-\delta)}{2}\,\zeta_*^2 .
\]
Summing over $t=0,\ldots,T-1$, using $L(\mathcal D,\boldsymbol\beta_T)\ge L^*$ and dividing by $\gamma\delta_M T/2$ gives item 2.
\end{enumerate}
\end{proof}

\noindent The following corollary applies Proposition~\ref{prop:transfer} to FRC, GFRC and BCGC.

\begin{corollary}\label{cor:topk}
Let items 1, 2 and 4 of Assumption~\ref{ass:transfer} hold, and let $0<\gamma\le1/\rho$.
\begin{enumerate}
\item If $\mathcal C$ is deterministic and $\delta$-contractive with $(1-\delta)(1+M)<1$, then FRC, GFRC and BCGC satisfy, under any pattern of at most $s$ stragglers and for every complete class decoded at each round,
\[
\frac1T\sum_{t=0}^{T-1}\|\mathbf g^{(t)}\|^2
\le\frac{2\big(L(\mathcal D,\boldsymbol\beta_0)-L^*\big)}{\delta_M\gamma T}+\frac{1-\delta}{\delta_M}\,\zeta_*^2,
\]
where $\delta_M=1-(1-\delta)(1+M)$.
\item In particular, if $\mathcal C=\mathrm{Top}\text{-}\kappa$ with $1\le\kappa\le d$ and $(d-\kappa)(1+M)<d$, then
\[
\frac1T\sum_{t=0}^{T-1}\|\mathbf g^{(t)}\|^2
\le\frac{2\big(L(\mathcal D,\boldsymbol\beta_0)-L^*\big)}{\delta_M\gamma T}+\frac{d-\kappa}{d\,\delta_M}\,\zeta_*^2,
\]
where $\delta_M=1-\frac{d-\kappa}{d}(1+M)$.
\end{enumerate}
\end{corollary}

\begin{proof}
\begin{enumerate}
\item Item 3 of Assumption~\ref{ass:transfer} holds. The encoding matrices of FRC, GFRC and BCGC are binary with nonzero rows, and the supports of each of their classes partition $[n]$ (Section~\ref{Sec_binary-gc}, Definition~\ref{def:gfrc}, and Proposition~\ref{partition1}). Since the $s+1$ classes are disjoint, at most $s$ stragglers leave a complete class $R_t\subseteq I_t$ at every round, and $\sum_{i\in R_t}\mathbf b_i=\one_{1\times n}$. Hence item 2 of Proposition~\ref{prop:transfer} applies.
\item $\mathrm{Top}\text{-}\kappa$ is deterministic and $\delta$-contractive with $\delta=\kappa/d$. The condition $(1-\delta)(1+M)<1$ then reads $(d-\kappa)(1+M)<d$, and $\frac{1-\delta}{\delta_M}=\frac{d-\kappa}{d\,\delta_M}$, so the claim follows from item 1.
\end{enumerate}
\end{proof}

By Corollary~\ref{cor:topk}, FRC, GFRC and BCGC have the same guarantee, and they differ only through the coarse partitions $\{\bar{\mathcal{D}}_i\}_{i\in R_t}$, that is, through the chunks that are compressed together. For $j\in\{0,\ldots,s\}$, let $\mathcal{P}_j:=\{S_i\}_{i\in K_j}$ be the partition of $[n]$ induced by the class $K_j$.

\noindent We now show that these $s+1$ partitions are genuinely different. The first step is that the cyclic supports make each class separate a distinct pair of adjacent chunks.

\begin{proposition}[Cyclic supports separate adjacent chunks]\label{prop:separation}
Let $\ell\ge2$, and consider the partitions $\mathcal{P}_j$ of BCGC, with indices modulo $n$. For every $l\in[n]$ and every $j\in\{0,\ldots,s\}$, the chunks $l$ and $l+1$ lie in different blocks of $\mathcal{P}_j$ if and only if $l+1\in K_j$. Hence exactly one class of BCGC separates the chunks $l$ and $l+1$.
\end{proposition}

\begin{proof}
Every support is a cyclic interval $S_i=\{i,\ldots,i+L_i-1\}$, and $L_i\le s+1+r<n$ since $\ell\ge2$, with $r:=0$ if $(s+1)\mid n$. If $l+1\in K_j$, the block of $\mathcal{P}_j$ containing $l+1$ is $S_{l+1}$, which starts at $l+1$ and does not contain $l$. Otherwise, the block $S_i$ of $\mathcal P_j$ containing $l+1$ satisfies $l+1=i+u$ with $1\le u\le L_i-1$, so $l\in S_i$. Since the classes partition $[n]$, $l+1$ belongs to exactly one class.
\end{proof}

\noindent This separation property makes the partitions pairwise distinct, so the decoded class genuinely matters under compression.

\begin{proposition}[Distinct partitions]\label{prop:diversity}
If $\ell\ge2$, the $s+1$ partitions $\mathcal{P}_0,\ldots,\mathcal{P}_s$ of BCGC are pairwise distinct, no one of them refines another, and $\mathcal P_j=\mathcal P_0+j$ when $(s+1)\mid n$. In contrast, all the classes of FRC induce the same partition, and those of GFRC induce one of the two partitions $\mathcal G_1$ and $\mathcal G_2$ of Definition~\ref{def:gfrc}.
\end{proposition}

\begin{proof}
Let $j\neq j'$. Since $j\in K_j$ and $j\notin K_{j'}$, Proposition~\ref{prop:separation} shows that $\mathcal{P}_j$ separates the chunks $j-1$ and $j$, whereas $\mathcal{P}_{j'}$ does not. Exchanging $j$ and $j'$, the partition $\mathcal{P}_{j'}$ separates the chunks $j'-1$ and $j'$, whereas $\mathcal{P}_j$ does not. Hence neither partition refines the other. If $(s+1)\mid n$, then $K_j=K_0+j$ (Proposition~\ref{Prop_Kshift}) and $S_{i+j}=S_i+j$, so $\mathcal P_j=\mathcal P_0+j$. Every class of $\mathbf{B}'_{\mathrm{FRC}}$ selects one row from each diagonal block, and the statement on GFRC follows from Definition~\ref{def:gfrc}.
\end{proof}

\section{Experiments}
\label{Sec_Experiments}

We implement all schemes in MPI (\texttt{mpi4py}) on the Toubkal supercomputer~\cite{kissami2025toubkal}, with one master and $n=100$ workers, each pinned to a core. Communication is modeled as $T=\alpha+\mathrm{bytes}/b$ ($\alpha=5\times10^{-4}$\,s, $b=30$\,MB/s), and the master's decoding time is included in every reported wall time. We use four datasets, \emph{Amazon} Employee Access with logistic regression ($d=213{,}776$), \emph{MNIST} and \emph{EMNIST} digits with softmax regression ($d=7{,}850$), and \emph{20~Newsgroups} with softmax regression on TF\nobreakdash-IDF features ($d=788{,}480$).

We compare BCGC with four schemes: the uncoded (\emph{naive}) scheme, the binary GFRC of~\cite{Charalambides2025}, which coincides with FRC when $(s+1)\mid n$, and two MDS gradient codes with real or complex coefficients, the Gaussian MDS code (G-MDS) of~\cite{Tandon2017} and the Reed--Solomon code (RS) of~\cite{Halbawi2018}. Each worker compresses its coded message with top $\kappa$ sparsification, $\kappa=\lfloor0.05\,d\rfloor$. The RS messages are complex, so we sparsify them by modulus and keep the real part of the decoded sum. All schemes run $100$ iterations of Nesterov accelerated gradient descent from $\boldsymbol\beta_0=\mathbf 0$ in double precision, with step size $\gamma=10$ on Amazon, $0.1$ on MNIST and EMNIST, and $1.0$ on 20~Newsgroups. This optimizer differs from the plain gradient descent of Proposition~\ref{prop:transfer}, so the experiments illustrate the analysis without being formally covered by it. Without compression, the naive and binary schemes recover the exact gradient and reach an AUC of $0.883$ on Amazon and an accuracy of $0.915$, $0.925$ and $0.617$ on MNIST, EMNIST and 20~Newsgroups. The MDS codes reach the same values at $s=20$. At $s=24$, however, the complex decoding of RS is numerically unstable, and its accuracy already falls to $0.02$ on 20~Newsgroups without compression.

All schemes share the same initialization, data and straggler schedule, so the differences between them come only from the coding scheme. Without compression, the runs of the exact schemes are identical, which confirms exact recovery, and top $\kappa$ sparsification is deterministic.

\subsection{Robustness of compressed gradient coding}
\label{Sec_ExpSetup}

We first isolate the effect of compression on the decoding step. Each scheme tolerates up to $s$ stragglers, and exactly $s$ workers are delayed at each iteration according to a fixed schedule. The binary schemes stop at the first complete class and sum its messages, whereas the MDS codes wait for $n-s$ responses and solve the linear system $\mathbf a_I\mathbf B=\one_{1\times n}$. Each worker compresses its message with top $\kappa$ sparsification, $\kappa=\lfloor0.05\,d\rfloor$. We consider two straggler levels, $s=20$, for which $(s+1)\nmid n$, and $s=24$, for which $(s+1)\mid n$.

Figure~\ref{fig:p1-amazon-metrics} shows the convergence on Amazon, Figure~\ref{fig:p1-wall} the wall time, and Tables~\ref{tab:stability}--\ref{tab:coeffs} the results on all four datasets. We report three findings: accuracy, wall time, and the effect of the decoding coefficients.

\emph{Accuracy.} Like the naive scheme, BCGC and GFRC stay within $0.022$ of their uncompressed metric. The MDS codes, in contrast, diverge under compression. At $s=20$ their metric falls to between $0.06$ and $0.58$. At $s=24$, G-MDS recovers the exact gradient without compression but diverges with it, while RS is numerically unstable even without compression, where its accuracy on 20~Newsgroups is already $0.02$; its apparent convergence under compression is therefore a numerical artifact rather than robustness.

\emph{Wall time.} BCGC and GFRC are the fastest coded schemes, about $27$\,s on Amazon against $46$ to $56$\,s for the MDS codes. The two binary schemes are identical at $s=24$. At $s=20$, where $(s+1)\nmid n$, BCGC is at most $17\%$ slower, on 20~Newsgroups, because of its larger load on a few workers (Remark~\ref{rem:loads}).

\emph{Decoding coefficients.} By Proposition~\ref{prop:diff}, the decoded error is bounded by $\|\mathbf a_I\|_1\max_i\|\boldsymbol\varepsilon_i\|$. The binary codes keep $\|\mathbf a_I\|_1\approx4$ to $5$ with unit coefficients, against about $10^{4}$ for G-MDS (Table~\ref{tab:coeffs}), which explains why the MDS codes amplify the compression error.

Overall, BCGC matches the binary state of the art and, unlike the MDS codes, preserves its accuracy under compression.
\begin{figure}[t]
\centering
\begin{subfigure}{\linewidth}
\centering
\includegraphics[width=\linewidth]{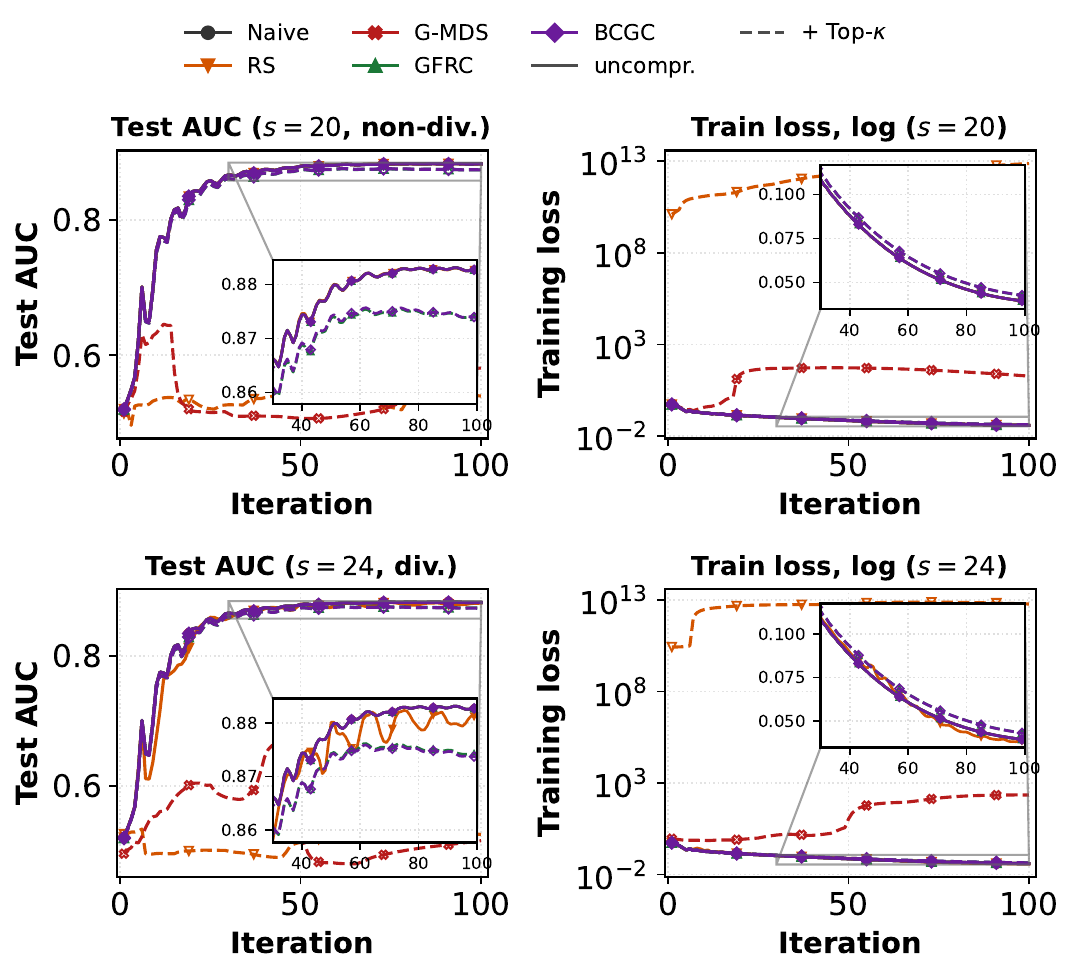}
\caption{Convergence: test AUC (left) and training loss on a logarithmic scale (right), for $s=20$ (top) and $s=24$ (bottom); solid lines uncompressed, dashed lines compressed.}
\label{fig:p1-amazon-metrics}
\end{subfigure}
\\[4pt]
\begin{subfigure}{\linewidth}
\centering
\includegraphics[width=\linewidth]{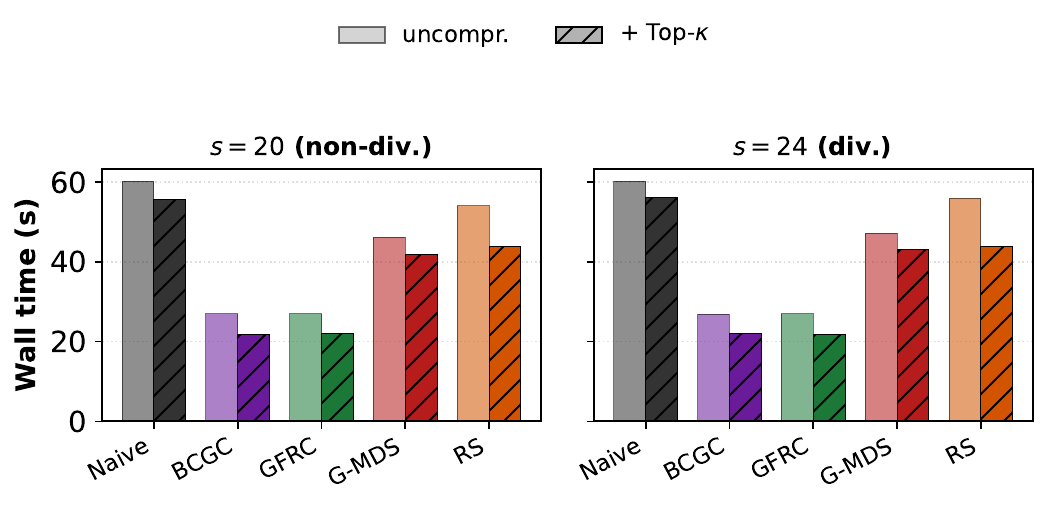}
\caption{Total wall time for $100$ iterations, without compression (plain bars) and with top $\kappa$ sparsification (hatched bars), for $s=20$ (left) and $s=24$ (right).}
\label{fig:p1-wall}
\end{subfigure}
\caption{Robustness on Amazon under top $\kappa$ sparsification, $n=100$.}
\label{fig:p1-amazon}
\end{figure}
\begin{table}[t]
\centering
\caption{Final test accuracy on all datasets with top $\kappa$ sparsification, $\kappa=\lfloor0.05\,d\rfloor$, $n=100$.}
\label{tab:stability}
\setlength{\tabcolsep}{3pt}\small
\begin{tabular}{llcccccc}
\toprule
Dataset & $s$ & Uncompr. & Naive & \textbf{BCGC} & GFRC & G-MDS & RS\\
\midrule
\multirow{2}{*}{Amazon} & 20 & \multirow{2}{*}{0.883} & 0.883 & 0.874 & 0.874 & 0.581 & 0.540\\
                        & 24 &                        & 0.883 & 0.874 & 0.874 & 0.516 & 0.525\\
\midrule
\multirow{2}{*}{MNIST}  & 20 & \multirow{2}{*}{0.915} & 0.895 & 0.894 & 0.893 & 0.196 & 0.087\\
                        & 24 &                        & 0.895 & 0.895 & 0.895 & 0.116 & 0.896\\
\midrule
\multirow{2}{*}{EMNIST} & 20 & \multirow{2}{*}{0.925} & 0.903 & 0.903 & 0.903 & 0.153 & 0.202\\
                        & 24 &                        & 0.903 & 0.903 & 0.903 & 0.206 & 0.904\\
\midrule
\multirow{2}{*}{20\,NG} & 20 & \multirow{2}{*}{0.617} & 0.617 & 0.609 & 0.609 & 0.151 & 0.063\\
                        & 24 &                        & 0.617 & 0.609 & 0.610 & 0.100 & 0.603\\
\bottomrule
\end{tabular}
\end{table}

\begin{table}[t]
\centering
\caption{Mean norm $\|\mathbf a_I\|_1$ of the decoding vector for $n=100$.}
\label{tab:coeffs}
\setlength{\tabcolsep}{3pt}\small
\begin{tabular}{llcccc}
\toprule
Dataset & $s$ & \textbf{BCGC} & GFRC & G-MDS & RS\\
\midrule
\multirow{2}{*}{Amazon} & 20 & $4.7$ & $4.7$ & $2.0{\times}10^{4}$ & $53$\\
                        & 24 & $4.0$ & $4.0$ & $1.3{\times}10^{4}$ & $2.2{\times}10^{2}$\\
\midrule
\multirow{2}{*}{MNIST}  & 20 & $4.8$ & $4.8$ & $2.0{\times}10^{4}$ & $53$\\
                        & 24 & $4.0$ & $4.0$ & $1.3{\times}10^{4}$ & $2.2{\times}10^{2}$\\
\midrule
\multirow{2}{*}{EMNIST} & 20 & $5.0$ & $4.9$ & $2.0{\times}10^{4}$ & $53$\\
                        & 24 & $4.0$ & $4.0$ & $1.3{\times}10^{4}$ & $2.2{\times}10^{2}$\\
\midrule
\multirow{2}{*}{20\,NG} & 20 & $4.5$ & $4.6$ & $2.0{\times}10^{4}$ & $53$\\
                        & 24 & $4.0$ & $4.0$ & $1.3{\times}10^{4}$ & $1.8{\times}10^{4}$\\
\bottomrule
\end{tabular}
\end{table}

\subsection{Partition diversity on heterogeneous data} \label{Sec_diversity}

All binary schemes share the same convergence guarantee and differ only in the partitions induced by their complete classes (Corollary~\ref{cor:topk}). To isolate this effect, we use a \emph{minimum error} decoder. Among the classes that are complete after $n-s$ responses, it selects the one with the smallest total compression error $\sum_{i\in K_j}\|\boldsymbol\varepsilon_i\|$, where $\boldsymbol\varepsilon_i=\mathcal C(\widehat{\mathbf g}_i)-\widehat{\mathbf g}_i$. Each worker also sends the scalar $\|\boldsymbol\varepsilon_i\|$, which adds an $O(1)$ communication cost, so BCGC and GFRC wait for the same $n-s$ responses and have the same wall time. We fix $s=24$, so that $(s+1)\mid n$, and use top $\kappa$ sparsification with $\kappa=\lfloor0.05\,d\rfloor$, $n=100$ and $b=30$\,MB/s. For each dataset we compare two data partitions. The i.i.d.\ partition is obtained by shuffling the data. The non-i.i.d.\ partition is obtained by sorting the data by label, so that each chunk contains essentially one class.

Table~\ref{tab:diversity} and Figures~\ref{fig:p2-mnist} and~\ref{fig:p2-acc} report the results. On the i.i.d.\ partition, BCGC and GFRC coincide, as predicted by the theory. On the non-i.i.d.\ partition, BCGC has a total compression error $15$ to $19\%$ lower than GFRC on Amazon, MNIST and EMNIST. It improves the accuracy by $5.0$ points on MNIST and $4.9$ points on EMNIST, removing about half of the accuracy lost to compression, at the same wall time of about $40$\,s. This wall time is longer than that of the first-complete-class decoder of Table~\ref{tab:stability}, which waits for fewer responses. On Amazon and 20~Newsgroups the two schemes agree to within $0.001$, because the partition has little effect when compression loses little information, as with the sparse text gradients, or when the chunks are similar. The gain therefore requires both heterogeneous data and a significant compression error, and BCGC is never worse than GFRC by more than $0.001$.

\begin{table}[t]
\centering
\caption{Partition diversity with the minimum error decoder and top $\kappa$ sparsification, $\kappa=\lfloor0.05\,d\rfloor$, $n=100$, $s=24$.}
\label{tab:diversity}
\setlength{\tabcolsep}{3pt}\small
\begin{tabular}{lcccccc}
\toprule
 & \multicolumn{2}{c}{i.i.d.} & \multicolumn{2}{c}{non-i.i.d.} & \multicolumn{2}{c}{$\sum_{i\in K}\|\boldsymbol\varepsilon_i\|$}\\
\cmidrule(lr){2-3}\cmidrule(lr){4-5}\cmidrule(lr){6-7}
Dataset & BCGC & GFRC & \textbf{BCGC} & GFRC & BCGC & GFRC\\
\midrule
Amazon (AUC) & 0.874 & 0.874 & \textbf{0.877} & \textbf{0.877} & $\mathbf{111}$ & $130$\\
MNIST        & 0.895 & 0.895 & \textbf{0.875} & 0.825 & $\mathbf{4.1{\times}10^{4}}$ & $5.0{\times}10^{4}$\\
EMNIST       & 0.902 & 0.903 & \textbf{0.880} & 0.831 & $\mathbf{1.9{\times}10^{5}}$ & $2.3{\times}10^{5}$\\
20\,NG       & 0.609 & 0.610 & \textbf{0.615} & 0.614 & $\mathbf{94}$ & $95$\\
\bottomrule
\end{tabular}
\end{table}

\begin{figure*}[t]
\centering
\includegraphics[width=\linewidth]{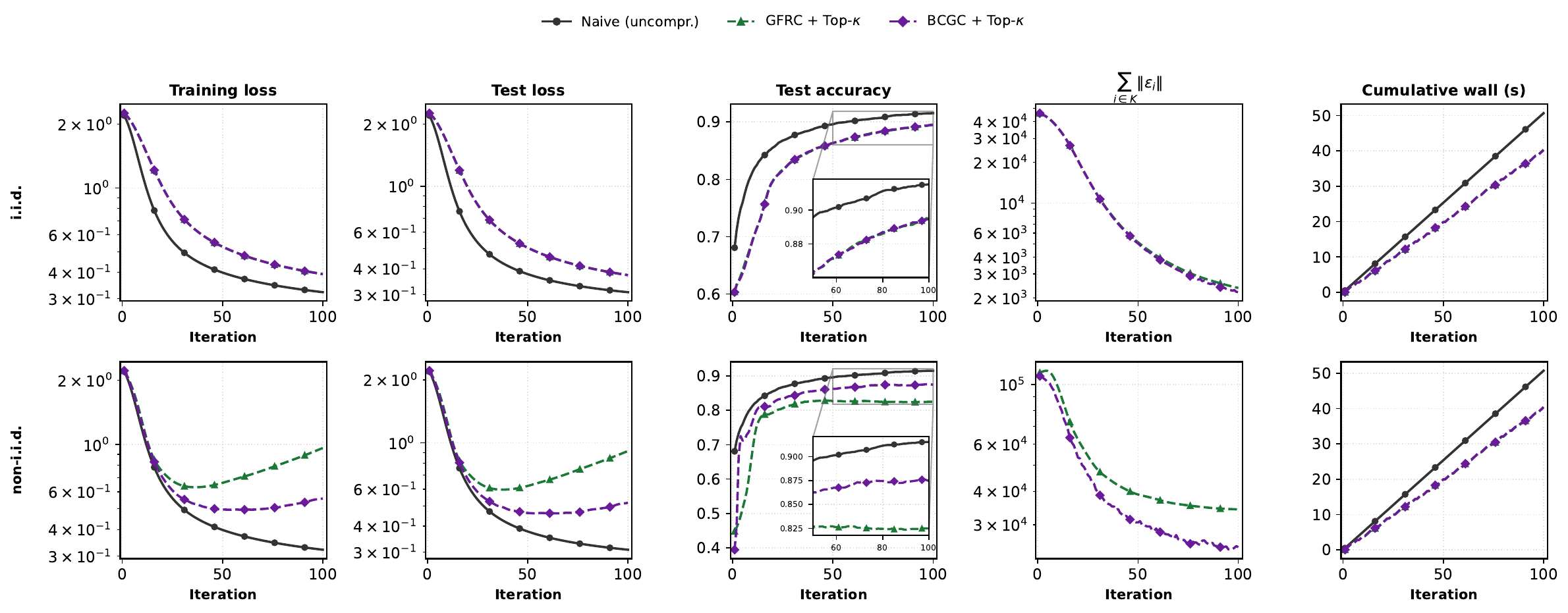}
\caption{Partition diversity on MNIST with the minimum error decoder, $n=100$, $s=24$, on the i.i.d.\ partition (top) and the non-i.i.d.\ partition (bottom).}
\label{fig:p2-mnist}
\end{figure*}

\begin{figure*}[t]
\centering
\includegraphics[width=\linewidth]{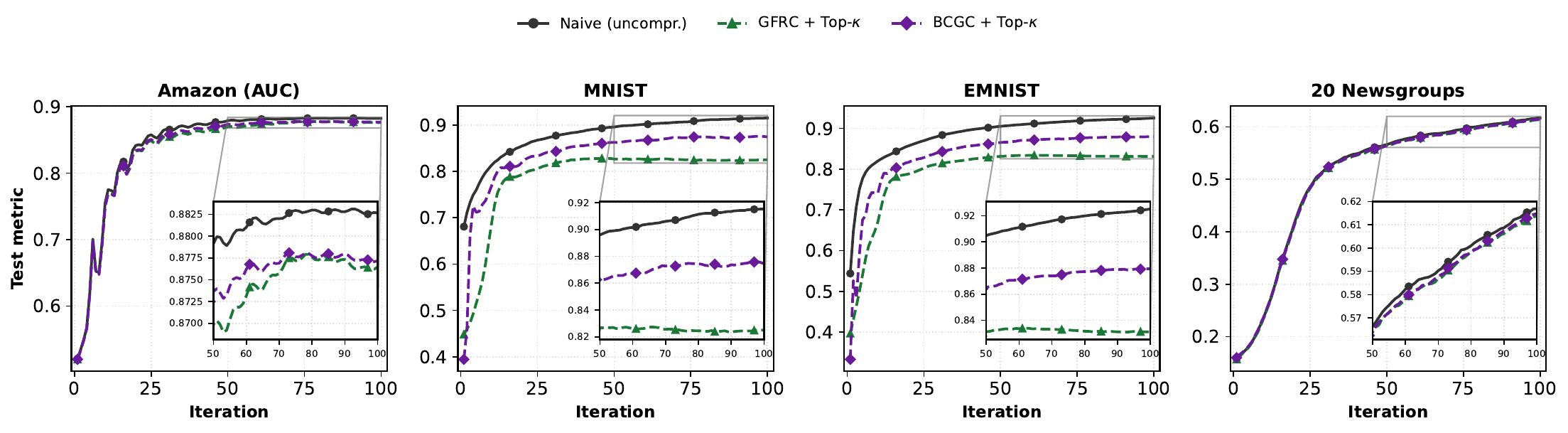}
\caption{Test metric on the non-i.i.d.\ partition with the minimum error decoder, $n=100$, $s=24$.}
\label{fig:p2-acc}
\end{figure*}

\section{Conclusion}
\label{Sec_Conclusion}

 We introduced BCGC, a binary gradient code with a circulant or near-circulant encoding matrix whose binary recovering sets are exactly its $s+1$ classes. It tolerates any $s$ stragglers at the optimal threshold $n-s$ with minimal total load $n(s+1)$, decodes by a single summation, and for $2\le s\le n-2$ is neither MDS nor permutation-equivalent to FRC or GFRC. Under compression it reduces to uncoded compression on a coarser partition, which keeps it stable where MDS codes amplify the error, and its $s+1$ distinct partitions let it select the reconstruction closest to the exact gradient. Experiments on four datasets confirm that BCGC matches GFRC, runs faster than the MDS codes, stays stable under compression, and improves accuracy by up to $5$ points on heterogeneous data. Future work will extend the approach to broader compressors and dynamic coding schemes.

\section*{Acknowledgment}
The authors gratefully acknowledge the support and computing resources provided by the Toubkal Supercomputer~\cite{kissami2025toubkal} (https://toubkal.um6p.ma) at UM6P, Morocco.

{\sloppy\printbibliography}

\end{document}